\documentclass[a4paper,USenglish]{lipics-v2016}
 
\usepackage{microtype}
\usepackage{colonequals}
\usepackage{bm}
\usepackage{tikz}
\usepackage{apxproof}
\usepackage{xcolor}
\usepackage{amssymb}
\usetikzlibrary{positioning, fit, calc, shapes, arrows}
\title{A Circuit-Based Approach to Efficient Enumeration}

\hypersetup{
    colorlinks,
    linkcolor={red!50!black},
    citecolor={blue!50!black},
    urlcolor={blue!30!black}
}

\author[1]{Antoine Amarilli}
\author[2]{Pierre Bourhis}
\author[3]{Louis Jachiet}
\author[4]{Stefan Mengel}
\affil[1]{LTCI, Télécom ParisTech, Université Paris-Saclay; France\\
  \texttt{antoine.amarilli@telecom-paristech.fr}}
\affil[2]{CRIStAL, CNRS UMR 9189 \& Inria Lille; France\\
  \texttt{pierre.bourhis@univ-lille1.fr}}
\affil[3]{Université Grenoble Alpes; France\\
  \texttt{louis.jachiet@inria.fr}}
\affil[4]{CNRS, CRIL UMR 8188; France\\
  \texttt{mengel@cril.fr}}

\Copyright{Antoine Amarilli; Pierre Bourhis; Louis Jachiet; Stefan Mengel}

\usepackage{chngcntr}

\theoremstyle{theorem}
\newtheoremrep{theorem}{Theorem}
\newtheoremrep{lemma}[theorem]{Lemma}
\newtheoremrep{proposition}[theorem]{Proposition}
\newtheoremrep{claim}[theorem]{Claim}
\newtheoremrep{remark}[theorem]{Remark}
\newtheoremrep{observation}[theorem]{Observation}

\counterwithin{theorem}{section}

\newcommand{\calA}{\mathcal{A}}
\newcommand{\calD}{\mathcal{D}}
\newcommand{\calE}{\mathcal{E}}
\newcommand{\calI}{\mathcal{I}}
\newcommand{\calO}{\mathcal{O}}
\newcommand{\calS}{\mathcal{S}}

\newcommand{\NN}{\mathbb{N}}

\renewcommand{\r}{\mathrm{r}}

\newcommand{\semp}{\{\}}
\newcommand{\ssemp}{\{\semp\}}

\newcommand{\var}{\mathrm{var}}

\newcommand{\ins}{\mathrm{in}}
\newcommand{\reach}{\mathrm{reach}}

\newcommand{\eqg}{=\!}
\newcommand{\geg}{\ge\!}
\newcommand{\bog}{\bowtie\!}

\newcommand{\hamw}[1]{\card{#1}}

\newcommand{\LVal}{\mathrm{LVal}}
\newcommand{\Leaf}{\mathrm{Leaf}}
\newcommand{\Assign}{\mathrm{Assign}}

\newcommand{\card}[1]{\left|{#1}\right|}

\begin{document}

\maketitle

\begin{abstract}
  We study the problem of enumerating the satisfying valuations of a circuit while
bounding the \emph{delay}, i.e., the time needed to compute each successive
valuation. We focus on the class of \emph{structured d-DNNF circuits} originally
introduced in knowledge compilation, a sub-area of artificial intelligence. We
propose an algorithm for these circuits that enumerates valuations with linear
preprocessing and delay linear in the Hamming weight of each valuation.
Moreover, valuations of constant Hamming weight can be enumerated with linear
preprocessing and constant delay.
 
Our results yield a framework for efficient enumeration that applies to all
problems whose solutions can be compiled to structured d-DNNFs. In particular,
we use it to recapture classical results in database theory, for factorized
database representations and for MSO evaluation. This gives an independent proof
of constant-delay enumeration for MSO formulae with first-order free variables
on bounded-treewidth structures.

\end{abstract}

\section{Introduction}
\label{sec:intro}
When a computational problem has a large number of solutions, 
computing all of them at once can take an unreasonable amount of time.
\emph{Enumeration algorithms} are an answer to this challenge, and have been
studied in many contexts (see \cite{Wasa16} for an overview).
They generally consist of two phases. First, in a \emph{preprocessing phase}, the
input is read and indexed.
Second, in an \emph{enumeration phase} that uses the result of the preprocessing, 
the solutions are computed one after the other. The goal is to limit the amount of time between each pair of successive
solutions, which is called \emph{delay}.

We focus on a well-studied class of efficient enumeration algorithms with very
strict requirements: the
preprocessing must be \emph{linear} in the input size, and the delay between
successive solutions must be \emph{constant}. Such algorithms have been
studied in particular
for database applications, to enumerate query answers
(see~\cite{DurandG07,bagan2006mso,DurandSS14,BaganDFG10,BaganDG07,kazana2013enumeration,KazanaS13}
and the recent survey~\cite{Segoufin14}), or to enumerate the tuples of
factorized database representations~\cite{olteanu2015size}.

One shortcoming of these existing enumeration algorithms is that they are typically
shown by building a custom index structure tailored to the 
problem, and
designing ad~hoc preprocessing and enumeration algorithms.
This makes it 
hard
to generalize 
them to other problems, or to
implement them efficiently. In our opinion, it would be far 
better if enumeration for multiple problems
could be 
done
via one generic representation of the results to enumerate,
reusing general algorithms for the preprocessing and enumeration phases.

This paper accordingly proposes a new framework for constant-delay enumeration algorithms,
inspired by the field of \emph{knowledge compilation} in artificial
intelligence. Knowledge compilation studies how the solutions to computational problems can
be \emph{compiled} to generic representations, in particular classes of \emph{Boolean
circuits}, on which reasoning tasks can then be solved using general-purpose algorithms. In this paper, we show how 
this knowledge compilation approach can be implemented for constant-delay enumeration, by compiling
to a
prominent class of circuits from knowledge compilation called \emph{deterministic decomposable negation normal form} (in short, d-DNNF)~\cite{darwiche2001tractable}.
These circuits 
generalize several forms of branching programs such as OBDDs~\cite{DarwicheM02}
and were recently shown to be more expressive than Boolean circuits of bounded treewidth~\cite{bova2017circuit}.
Further, there are many efficient algorithms to compute d-DNNF representations
of small width CNF formulae for a wide range of width notions~\cite{BovaCMS15},
and even software implementations to compute such representations for given
Boolean functions~\cite{OztokD15,ChoiD13}. d-DNNFs are also intimately related
to state-of-the-art propositional model counters based on exhaustive
DPLL~\cite{HuangD05}, to syntactically multi-linear arithmetic
circuits~\cite{RazSY08}, and to probabilistic query evaluation in
database theory~\cite{JhaS13}.

Our main technical contribution is an efficient algorithm to enumerate the
satisfying valuations of 
a d-DNNF under a standard structuredness assumption, namely, assuming that a
so-called \emph{v-tree} is given~\cite{PipatsrisawatD08}: this assumption holds in all works cited above. 
Our first main result (Theorem~\ref{thm:main}) shows that we can enumerate the satisfying valuations of such a circuit
with linear preprocessing and delay linear in the Hamming weight of each
valuation. Further, our second main result (Theorem~\ref{thm:constant}) shows
that, if we impose a constant bound on the Hamming weight, we can
enumerate the valuations with constant delay. In these results we express valuations succinctly as the set of
the variables that they set to true.

To show our results, we consider d-DNNFs under a semantics where negation is
implicit, i.e., variables that are not tested must be set to zero. 
In analogy to zero-suppressed OBDDs~\cite{wegener2000branching}, we call this semantics \emph{zero-suppressed}. The preprocessing phase of our algorithm rewrites such circuits to
a normal form (Section~\ref{sec:normal}) and pre-computes a multitree
reachability index on them (Section~\ref{sec:multitrees}), which allows us to
enumerate efficiently the \emph{traces} of the circuit, and obtain the desired valuations
(Section~\ref{sec:enum}).
To enumerate for d-DNNFs in standard semantics, we show how to rewrite the input circuit to zero-suppressed semantics, using
the structuredness assumption, and using a new notion of \emph{range gates} to make the process efficient
(Section~\ref{sec:reducing}). The overall proof is very modular; for an outline see Figure~\ref{fig:schema}.
\begin{figure}[t]
  \begin{tikzpicture}[yscale=.48, every node/.style={inner sep=2pt}]

\node[draw] (ddnnf) at (0, .8) {d-DNNF};

\node[draw] (vtree) at (-.22, -.8) {v-tree};

\node[draw,align=center] (zss) at (2.4, 0)
    {augmented\\[-.2em] d-DNNF\\[-.2em]
{\footnotesize (Def.~\ref{def:zss})}};

\node[draw,align=center] (nf) at (4.8,0)
    {normal\\[-.2em] d-DNNF\\[-.2em]
{\footnotesize (Def.~\ref{def:normalform})}};

\node[draw,align=center] (indx) at (7.3,0)
    {normal\\[-.2em] d-DNNF\\[-.2em]+OR-index};

\node[draw,align=center] (ctra) at (9.85,0)
    {compressed\\[-.2em] traces\\[-.2em]
{\footnotesize (Def.~\ref{def:comptrace})}};

\node[draw,align=center] (val) at (12.33,0)
    {satisfying\\[-.2em] valuations};

\node[align=center] (redlabel) at (1.15,.05)
{\footnotesize Prp.\\[-.3em]\ref{prp:reducing}};

\draw[->] (vtree) -- (vtree-|zss.west);
\draw[->] (ddnnf) -- (ddnnf-|zss.west);
\draw[->] (zss) -- node[above] {\footnotesize Prp.}
node[below] {\footnotesize \ref{prp:normalize}} (nf);
\draw[->] (nf) -- node[above] {\footnotesize \vphantom{p}Thm.}
node[below]{\footnotesize \ref{thm:orindex}} (indx);
\draw[->] (indx) -- node[above]
{\footnotesize Prp.} node[below]{\footnotesize\ref{prp:enumtraces}} (ctra);
\draw[->] (ctra) -- node[above]
{\footnotesize Prp.} node[below]{\footnotesize\ref{prp:enumsol}} (val);

\draw (-.8,1.65) -- (-.8,1.9) -- 
node[above] {\emph{Linear-time preprocessing phase
(Sec.~\ref{sec:reducing}--\ref{sec:multitrees})}}
(8.25, 1.9) -- (8.25, 1.65);

\draw (8.35, 1.65) -- (8.35, 1.9) -- 
node[above] {\emph{Enumeration phase
(Sec.~\ref{sec:enum})}}
(13.25, 1.9) -- (13.25, 1.65);

\end{tikzpicture}
  \caption{Overview of the proof of Theorem~\ref{thm:main}}
\label{fig:schema}
\end{figure}
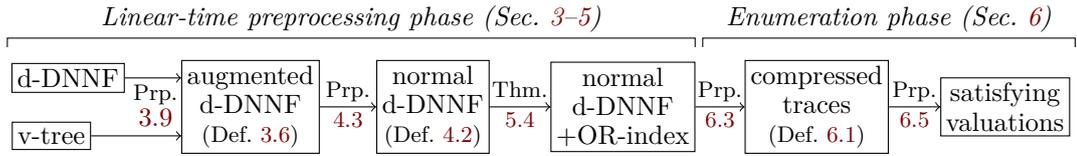

Our second contribution is to illustrate how our circuit-based
framework and enumeration results can be useful in database theory.
As a proof of concept, we present two known results that we can extend, or recapture with
an independent proof.
First,
we re-prove with our framework that the answers to MSO queries on trees and
bounded-treewidth structures can be enumerated with linear preprocessing and delay
linear in each assignment, i.e., constant-delay if the free variables are
first-order. This was previously shown by Bagan~\cite{bagan2006mso} with a
custom construction, by Kazana and Segoufin~\cite{kazana2013enumeration}
using a powerful result of Colcombet~\cite{colcombet2007combinatorial}, and by
Courcelle~\cite{courcelle2009linear} in a more general setting (but with
$O(n \log n)$ preprocessing) using \mbox{AND/OR-DAGs} (that share some
similarities with DNNFs). 
Our proof follows our proposed approach: we compute a circuit representation of
the output following the provenance constructions
in~\cite{amarilli2015provenance}, and simply apply our enumeration result to this
circuit.
Second, we show how d-DNNFs generalize the deterministic factorized
representations of relational instances studied in database theory~\cite{olteanu2015size}.
This allows us to give enumeration algorithms with linear preprocessing and
constant delay for arbitrary deterministic d-representations, extending the
enumeration result of~\cite{olteanu2015size}. 

The paper is structured as follows. Section~\ref{sec:prelim} gives the main
definitions and results. We then describe the preprocessing phase of our
algorithm: we reduce the input circuit to
zero-suppressed semantics in Section~\ref{sec:reducing}, rewrite it to a normal
form in Section~\ref{sec:normal}, and compute the multitree index in
Section~\ref{sec:multitrees}. We then describe the enumeration algorithm in
Section~\ref{sec:enum}. We present our two applications in
Section~\ref{sec:applications} and conclude in Section~\ref{sec:conclusion}. Due to space restrictions, many details and the proofs are found in the appendix.

\section{Preliminaries and Problem Statement}
\label{sec:prelim}
\subparagraph*{Circuits.}
A circuit $C = (G, W, g_0, \mu)$ 
is a directed acyclic graph $(G, W)$ whose vertices $G$
are called \emph{gates}, whose edges $W$ are called \emph{wires}, which has an
\emph{output gate} $g_0 \in G$, and where each gate $g \in G$ has a
\emph{type} $\mu(g)$ among $\land$ (AND-gate), $\lor$ (OR-gate), $\neg$
(NOT-gate), or $\var$ (variable).
We represent the circuit with adjacency lists 
that indicate, for each gate $g \in G$, the gates having a wire to~$g$ (called the
\emph{inputs} of~$g$),
and the gates
of which $g$ is an input;
the number of such gates is called respectively the \emph{fan-in} and
\emph{fan-out} of~$g$. The size $\card{C}$ of this 
representation
is then $\card{G} + \card{W}$.
We require that variables have fan-in zero, that NOT-gates have
fan-in one, and we will always work on \emph{negation normal form} (NNF)
circuits where 
the input of NOT-gates is always a variable. 
A circuit without NOT-gates is called
\emph{monotone}.

We write $C_\var$ for the set of variables of~$C$.
A \emph{valuation} of~$C_\var$ is a function $\nu: C_\var \to \{0, 1\}$. A
circuit defines a \emph{Boolean function} on~$C_\var$, that is, a function
$\phi$ that maps each valuation of~$C_\var$ to $\{0, 1\}$. For any
valuation~$\nu$, the image of~$\nu$ by~$\phi$ is defined by substituting each
gate in~$C_\var$ by its value according to~$\nu$, evaluating the circuit using
the standard semantics of Boolean operations, and returning the value of the
output gate~$g_0$.
Note that AND-gates (resp., OR-gates) with no inputs always evaluate to~$1$
(resp., to~$0$) in this process. We call a gate \emph{unsatisfiable} if it
evaluates to~$0$ under all valuations (and \emph{satisfiable} otherwise); we
call it \emph{0-valid} if it evaluates to~$1$
under the valuation which sets all variable gates to~$0$.
We say that $\nu$ \emph{satisfies} $C$ 
if $\phi$ maps~$\nu$ to~$1$ (i.e., $g_0$ evaluates to~$1$ under~$\nu$), and
call $\nu$ a \emph{satisfying valuation}.

For enumeration, we represent a valuation $\nu$ of~$C$ as the set $S_\nu$ of variables
of~$C_\var$ that it sets to~$1$, i.e., $\{g \in C_\var \mid \nu(g) = 1\}$.
We call $S_\nu$ an \emph{assignment},
and a \emph{satisfying assignment} if~$\nu$ is a satisfying valuation.
The \emph{Hamming weight} $\hamw{\nu}$ of~$\nu$ is the cardinality
of~$S_\nu$. 
Unlike valuations, assignments of
constant Hamming weight are of constant size, no matter the size of~$C_\var$.
We write $\semp$ for the empty assignment, and write $\emptyset$ for an empty
set of assignments.

The main class of circuits that we will study are \emph{d-DNNFs}
\cite{darwiche2001tractable}, of which we now recall the definition.
We say that an AND-gate $g$ of a circuit~$C$ is
\emph{decomposable} if there is no pair $g_1 \neq g_2$ of input gates to~$g$
such that some variable $g' \in C_\var$ has a directed path both to~$g_1$ and
to~$g_2$: intuitively, a decomposable AND-gate is a conjunction of inputs on disjoint
sets of variables. We say that an OR-gate $g$ of~$C$ is \emph{deterministic} if there is
no pair $g_1 \neq g_2$ of input gates of~$g$ and valuation $\nu$ of~$C$ such
that $g_1$ and $g_2$ both evaluate to~$1$ under~$\nu$: intuitively, a
deterministic OR-gate is a disjunction of mutually exclusive inputs. A circuit~$C$ is a
\emph{d-DNNF} if all its AND-gates are decomposable, and all its OR-gates are
deterministic.

We will further study the subclass of d-DNNFs called \emph{structured d-DNNFs},
consisting of the d-DNNFs having a \emph{v-tree} \cite{PipatsrisawatD08}.
A \emph{v-tree} on a set $S$ of variables is a rooted unranked ordered
tree $T$ whose set of leaves is exactly $S$. We write $<_T$ for the order on~$T$
in which the nodes are visited in a pre-order traversal. For a circuit $C$, we
say that a v-tree~$T$ on the set~$C_\var$ is a \emph{v-tree}
of~$C$ if there is a mapping $\lambda$ from the gates of~$C$ to the nodes of~$T$
such that:
(i) $\lambda$ maps the variables of~$C$ to themselves; 
(ii) for each 
wire $(g, g')$ of~$C$, the node $\lambda(g)$ is a descendant of~$\lambda(g')$ in~$T$; and
(iii) for each AND-gate $g$ of~$C$ with inputs $g_1,
    \ldots, g_n$ (in this order),
the nodes $\lambda(g_1), \ldots, \lambda(g_n)$ are descendants
  of~$\lambda(g)$, none of them
  is a descendant of another, and we have $\lambda(g_1) <_T \cdots <_T
  \lambda(g_n)$.
Note that having a v-tree implies (by point~iii) that all AND-gates are 
decomposable. A \emph{structured d-DNNF} is a d-DNNF $C$ given with a
v-tree $T$ of~$C$.

\subparagraph*{Enumeration.}
As usual for efficient enumeration algorithms~\cite{Segoufin14}, we work
in the RAM model with uniform cost measure (see,
e.g.,~\cite{AhoHU74}), 
where pointers, numbers, labels for vertices and edges, etc.,
have constant size; thus an assignment has size linear in its Hamming weight.

An \emph{enumeration algorithm with linear-time preprocessing} computes a
set of results $\calS(\calI)$ from an input instance $\calI$. It consists of two
parts.
First, the \emph{preprocessing phase} 
takes as input
an instance $\calI$ and produces in \emph{linear time} an indexed
instance $\calI'$ and an initial \emph{state}. 
Second, 
the \emph{enumeration phase}
repeatedly calls an algorithm $\calA$. Each call to~$\calA$ takes as input the
indexed instance $\calI'$ and the current state,
and returns a result and a new state: a special state value indicates that
the enumeration is over so $\calA$ should not be called again.
The results produced by the calls to~$\calA$
must be exactly the elements of $\calS(\calI)$, with no
duplicates.

We say that the enumeration
algorithm has \emph{linear delay} if the
time to produce each new output element $\calE$ is linear in the size of $\calE$
(and independent of the input instance $\calI$).
In particular, when the output elements have constant size, 
each element 
must be produced with constant delay, which we call \emph{constant-delay
enumeration}.
The \emph{memory usage} of an enumeration algorithm is the maximum number of
memory cells used during the 
enumeration phase (not counting the indexed instance $\calI'$, which resides in
read-only memory), expressed
as a function of
the input instance size $\card{\calI}$ and of the size $\card{\calO}$ of the
largest output (as in~\cite{bagan2006mso}).
Note that constant delay does not imply a bound on memory usage, because
the state can become large even if we only add a constant
quantity of information at each step.

\subparagraph*{Main results.}
Our main theorem on circuit enumeration is the following:

\begin{theorem}
  \label{thm:main}
  Given a structured d-DNNF $C$ with a v-tree $T$, we
  can enumerate its satisfying assignments with linear-time preprocessing, linear
  delay, and memory usage \mbox{$O(\card{\calO} \cdot \log \card{C})$},
  where $\card{\calO}$ is the Hamming weight of the largest assignment.
\end{theorem}

If we fix a maximal Hamming weight $k \in \NN$, we can show constant-delay
enumeration:

\begin{theorem}
  \label{thm:constant}
  For any $k \in \NN$, 
  given a structured d-DNNF $C$ with a v-tree $T$,
  we can enumerate its satisfying assignments of Hamming weight $\leq k$ with
  preprocessing in time $O(\card{T} + k^2 \cdot \card{C})$, delay in $O(k)$, 
  and memory in $O(k \cdot \log \card{C})$,
  i.e., linear-time preprocessing and constant delay for fixed~$k$.
\end{theorem}

In both results,
remember that $\card{C}$ is the number of gates \emph{and wires} of~$C$.
We prove our two results in Sections~\ref{sec:reducing}--\ref{sec:enum}: the first
three sections present the three steps of the linear-time preprocessing algorithm,
and the last one presents the enumeration algorithm.
We then use the results for database applications in
Section~\ref{sec:applications}, in particular 
re-proving constant-delay enumeration for MSO queries with free first-order
variables on bounded-treewidth structures.

The memory bound in our results is not constant and depends
logarithmically on the input. While we think that this is reasonable,
we also show constant-memory enumeration for some restricted
circuit classes: the details are deferred to Appendix~\ref{apx:memory} for lack
of space.

\section{Reducing to Zero-Suppressed Semantics}
\label{sec:reducing}
\begin{toappendix}
  \label{apx:reducing}
  We now introduce some additional notation to be used throughout the appendix.
We will call a \emph{decomposable circuit} or \emph{DNNF} a
circuit where all AND-gates are decomposable, and a \emph{deterministic
circuit} a circuit where all OR-gates are deterministic.

We introduce an additional notation related to zero-suppressed semantics:
\begin{definition}
\label{def:capset}
For a gate $g$ in a monotone augmented circuit~$C$, we will call its
  \emph{captured set} $S(g)$ the set of the minimal valuations of~$g$ (written
  as assignments)
  in the sense of Definition~\ref{def:minval} (and for which an alternative
  characterization is given as Lemma~\ref{lem:botup}). The captured set $S(g_0)$ of the
  output gate~$g_0$ of~$C$ is then equal to its set of assignments $S(C)$; so we
  may also call $S(C)$ the \emph{captured set} of~$C$.
\end{definition}

We also introduce an additional definition related to partial traces (in
particular, traces):

\begin{definition}
  \label{def:tested}
  The variables \emph{tested} by a partial trace~$T$ of~$C$ are the variables that
  occur in~$T$ or occur in the interval of a range gate of~$T$. Note that
  these are a subset of the interval of the root of~$T$.
\end{definition}

\end{toappendix}

We start our linear preprocessing by rewriting the input circuit to an
alternative \emph{zero-suppressed} semantics
where negation is coded implicitly.
For this rewriting, we will use the structuredness assumption on the circuit, in a weaker
form called having a \emph{compatible
order}: this is the first thing that we present. We will also extend
slightly our circuit formalism, to concisely represent sets of inputs with
\emph{range gates} that use this order: we present this second. Last, we present the alternative
semantics, and give our translation result (Proposition~\ref{prp:reducing}).

\subparagraph*{Compatible orders.}
Our structuredness requirement is to have a
\emph{compatible order}:

\begin{definition}
  \label{def:compatible}
  An \emph{order} for a circuit~$C$ is a total order $<$ on~$C_\var$.
  For two
  variables $g_1, g_2 \in C_\var$, the \emph{interval} $[g_1, g_2]$ consists of
  the variables $g$ which are between $g_1$ and~$g_2$ for~$<$, i.e., $g_1 \leq g \leq g_2$
  or $g_2 \leq g \leq g_1$.
  The \emph{interval} of a gate~$g$ is then $[\min(g), \max(g)]$, where $\min(g)$
  denotes the smallest gate according to~$<$ that has a directed path to~$g$, 
  and $\max(g)$ is defined analogously.
  In particular, the interval of any $g \in C_\var$ is $[g, g] = \{g\}$.
  
  We say that the
  order $<$ is \emph{compatible} with $C$ if, for every AND-gate $g$ with inputs
  $g_1, \ldots, g_n$ (in this order), for all $1 \leq i < j \leq n$, we have
  $\max(g_i) < \min(g_j)$; in particular, the intervals of~$g_1, \ldots, g_n$ are pairwise disjoint.
\end{definition}

Note that, if a circuit $C$ has compatible order $<$, every AND-gate $g$
is decomposable: if some $g' \in C_\var$ had a directed path
to two inputs of~$g$ then their intervals would intersect.

Observe further that, given a structured d-DNNF $C$ with a v-tree~$T$, we 
can easily compute a compatible order~$<$ for~$C$
in linear time in~$T$. Indeed, let $<$ be the restriction to~$C_\var$ of the order~$<_T$ on~$T$
given by pre-order traversal.
Considering any suitable mapping~$\lambda$ from~$C$ to~$T$, for any gate~$g$,
we know that $\min(g)$ is
no less than the first leaf of~$T$ in~$<$ reachable from~$\lambda(g)$, and that
$\max(g)$ is no greater than the last leaf reachable from~$\lambda(g)$.
The
intervals of the inputs $g_1, \ldots, g_n$ to an AND-gate are then pairwise
disjoint, because they are
included in the sets of reachable leaves from the nodes $\lambda(g_1), \ldots,
\lambda(g_n)$ in the
v-tree, and none of these nodes is a
descendant of another, so they cannot share any descendant leaf. Hence, 
if we know a \mbox{v-tree}~$T$ for~$C$
then we know an order~$<$ for~$C$.

\subparagraph*{Augmented circuits.}
We use compatible orders to 
define circuits with a new type of gates:

\begin{definition}
  \label{def:augmented}
  For $k \in \NN$, we define a \emph{$k$-augmented circuit} $C$ as a circuit
  with a compatible order~$<$ and with
  $k$ additional types of gates, called \emph{range gates}: there are the
  \emph{${\eqg i}$-range gates} for $0 \leq i < k$, and the \emph{${\geg k}$-range
  gates}. These gates must have exactly two inputs, which must be variables
  of~$C$ (they are not necessarily different, so we allow multi-edges in
  circuits for this purpose).
  We talk of \emph{augmented circuits} when the value of~$k$
  does not matter.

  When evaluating a $k$-augmented circuit under a valuation $\nu$, each ${\eqg
  i}$-range gate~$g$ (resp., ${\geg k}$-range gate $g$)
  with inputs $g_1$ and $g_2$ evaluates to~$1$ if there
  are exactly $i$ gates (resp., at least~$k$ gates) in $[g_1, g_2]$ set to~$1$
  by~$\nu$; note that $g$ may be unsatisfiable if $\card{[g_1, g_2]}$ is
  too small.
\end{definition}

Range gates are related to the threshold gates studied in circuit complexity (see
e.g.~\cite{BarringtonIS90}), but we only apply them directly to variables.
We can of course emulate range gates with standard gates, 
e.g.,
${\geg 0}$-gates always
evaluate to~$1$, and
a ${\geg 1}$-range gate on $g_1$ and $g_2$ can be expressed
as an OR-gate $g$ having the interval $[g_1, g_2]$ as its set of inputs. However, the point of range gates
is that we can now write this in constant space, thanks to~$<$. This will be
important to rewrite circuits in linear time to our alternative semantics.

\subparagraph*{Zero-suppressed semantics.}
We are now ready to introduce our alternative semantics for 
augmented circuits. We will do so only on \emph{monotone} augmented circuits,
i.e., without NOT-gates, because
negation will be coded implicitly. 
We use the notion of \emph{traces}:

\begin{definition}
  \label{def:trace}
  An \emph{upward tree} $T$ of a monotone augmented circuit $C = (G, W, \mu,
  g_0)$ is a subgraph $(G', W')$ of~$C$, with $G' \subseteq G$ and $W' \subseteq
  W$, which is a rooted tree up to reversing the direction of the wires.
  For all $(g', g) \in W'$, we 
  call $g' \in G'$ a \emph{child} of $g \in G'$ in~$T$, and call~$g$ the
  \emph{parent} of~$g'$ in~$T$;
  note that $g'$
  is an input of~$g$ in~$C$. A gate $g \in G'$ in~$T$ is an \emph{internal gate}
  of~$T$ if it
  has a child in~$T$, and a \emph{leaf} otherwise.
  $T$ is a \emph{partial trace} if its internal gates are
  AND-gates and OR-gates and if its gates satisfy the following:
  \begin{itemize}
    \item for every AND-gate $g$ in~$T$, \emph{all} its
      inputs in~$C$ are children of~$g$ in~$T$;
    \item for every OR-gate $g$ in~$T$, \emph{exactly one} of its inputs in~$C$ is a child
      of~$g$ in~$T$.
  \end{itemize}
  Note that $T$ cannot contain OR-gates with no inputs, and that its leaves
  consist of
  range gates, variable gates, and AND-gates with no inputs.
  We call $T$ a \emph{trace} of~$C$ if its root is~$g_0$.
\end{definition}

We define traces as trees, not general DAGs, because we cannot reach the same
gate in a trace by two different paths (remember that AND-gates in augmented circuits are
decomposable).
We can see each trace $(G', W')$ of~$C = (G, W, \mu, g_0)$ as an augmented
circuit $(G', W', \mu, g_0)$,
up to adding to range gates in the trace their inputs in~$C$,
and we then have:

\begin{observation}\label{obs:trace}
 A valuation $\nu$ of a monotone augmented circuit $C$ satisfies $C$ if and only if $\nu$ satisfies a trace of $C$.
\end{observation}

Observe that we can check if a valuation $\nu$ of~$C$ satisfies a trace $T$
simply by looking at the value of~$\nu$ on the leaves of~$T$; the
definition of~$\nu$ outside the intervals of the leaves does not matter.
We will change this point to define zero-suppressed semantics, 
where $\nu$ can only satisfy~$T$ if it maps to~$0$ all the other variables. We
then call $\nu$ a \emph{minimal valuation} of~$T$:

\begin{definition}
  \label{def:minval}
  Let $C$ be a monotone augmented circuit, $\nu$ be a valuation of~$C$, and $T$
  be a trace or partial trace of~$C$. We call $\nu$ a \emph{minimal
  valuation} of~$T$ if:
  \begin{itemize}
    \item For every variable $g$ in~$T$, we have $\nu(g) = 1$;
    \item For every ${\bog i}$-range gate $g$ in~$T$ with inputs $g_1$ and $g_2$
      in~$C$ (where ${\bowtie} \in \{{=}, {\geq}\}$ and $i \in \NN$), the number $n$ of variables in~$[g_1, g_2]$ that are set to~$1$
      by~$\nu$ satisfies the constraint $n \bowtie i$;
    \item All other variables of~$C_\var$ are set to~$0$ by~$\nu$.
  \end{itemize}
  Note that this implies that $\nu$ satisfies~$T$. We call $\nu$ a
  \emph{minimal valuation} for a gate $g$ of~$C$ (resp., for~$C$) if it is a
  minimal valuation of a partial trace rooted at~$g$ (resp., at the output~$g_0$).
\end{definition}

Note that $C$ may have two minimal valuations $\nu_1$ and~$\nu_2$ whose assignments
$S_1$ and $S_2$ are such that $S_1 \subsetneq S_2$ (see, e.g.,
Example~\ref{exa:circuit} below). Minimality only imposes that, relatively to a
trace $T$, the valuation sets to~$0$ all variables that are not tested in~$T$.
Minimal valuations allow us to define the \emph{zero-suppressed semantics} of a monotone augmented circuit $C$:
the satisfying valuations of~$C$ in this semantics are those that are minimal for some trace. 

\begin{definition}
  \label{def:zss}
  \label{def:zssddnnf}
  A monotone augmented circuit $C$ in \emph{zero-suppressed semantics}
  captures the (generally non-monotone) Boolean function $\Phi$ 
  mapping a valuation~$\nu$ to~$1$ iff $\nu$ is a minimal
  valuation for~$C$. We call $S(C)$ the set of satisfying assignments
  of~$C$ in this semantics.

  We call $C$ a \emph{d-DNNF in zero-suppressed semantics} if it satisfies
  the analogue of determinism: there is no OR-gate $g$ with two inputs $g_1 \neq
  g_2$ and valuation $\nu$ of~$C$ that is a minimal valuation for both~$g_1$
  and~$g_2$. (Decomposability again follows from the compatible order.)
\end{definition}

\begin{example}
\label{exa:circuit}
  Consider the monotone circuit $C$ whose output gate is an OR-gate with three
  inputs: $x$,
  $y$, and an AND-gate of~$y$ and~$z$. The circuit $C$
  captures $x \vee y$ in
  standard semantics, and it is not a d-DNNF. $C$ has three traces, having one
  minimal valuation each. In the zero-suppressed semantics, we have $S(C) = \{\{x\}, \{y\}, \{y, z\}\}$, and
  $C$ captures the Boolean function $(x
  \land \neg y \land \neg z) \lor (\neg x \land y)$. Further, $C$ is a d-DNNF in
  that semantics.
\end{example}

Zero-suppressed semantics makes enumeration easier, because it expresses
negation implicitly in a very concise way. The name is
inspired by zero-suppressed OBDDs~\cite[Chapter~8]{wegener2000branching}: variables
that are not tested when following a trace are implicitly set to~$0$. 
We can equivalently define the assignments $S(C)$ of~$C$ inductively as follows:

\begin{toappendix}
  We then show the auxiliary characterization of the set of assignments, which
  we will use heavily in the proofs:
\end{toappendix}

\begin{lemmarep}
  \label{lem:botup}
  Let $C$ be a monotone augmented circuit. Let us define inductively a set of assignments
  $S(g)$ for each gate~$g$ in the following way:
\begin{itemize}
\item for all $g \in C_\var$, we set $S(g) \colonequals \{g\}$;
\item for all ${\bog i}$-range gates $g$ with inputs $g_1$ and $g_2$, we
  set 
   $S(g) \colonequals \{t \subseteq [g_1, g_2] \mid \card{t} \bowtie i\}$;
\item for all OR-gates $g$ with inputs $g_1, \ldots, g_n$, we set $S(g) \colonequals
  \bigcup_{1 \leq i \leq n} S(g_i)$ (with $S(g) = \emptyset$ if $g$ has no
    inputs);
  \item for all AND-gates $g$ with inputs $g_1, \ldots, g_n$, we set
    $S(g) \colonequals \{S_1 \cup \cdots \cup S_n \mid (S_1, \ldots, S_n) \in
    \prod_{1 \leq i \leq n} S(g_i)\}$ (with $S(g) = \ssemp$ if $g$ has no
    inputs); observe that the unions are always disjoint because $C$ has a
    compatible order.
\end{itemize}
  Then, for any gate $g$, the set $S(g)$ contains exactly the assignments
  that describe a minimal valuation for~$g$. In particular, for $g_0$ the
  output gate of~$C$, the set $S(g_0)$ is exactly $S(C)$.
\end{lemmarep}

\begin{proof}
  We show the claim by induction.
  For the base cases:
  \begin{itemize}
    \item For a variable $g$, the only partial trace rooted at~$g$ is $\{g\}$, and
      indeed its only minimal valuation is~$\{g\}$.
    \item For a range gate $g$ with inputs $g_1$ and $g_2$, the only partial trace
      rooted at~$g$ is $\{g\}$, and its minimal valuations are as defined.
  \end{itemize}

  For the induction cases:
  \begin{itemize}
    \item For an OR-gate $g$, if $g$ has no inputs, then there is no partial trace
      rooted at $g$, so $S(g) = \emptyset$ is correct. If $g$ has inputs, then
      we can partition the partial traces rooted at $g$ depending on which input is
      retained. In particular, the set of leaves of the partial traces rooted at~$g$ are
      exactly the union of the set of leaves of the partial traces rooted at the inputs
      of~$g$. Hence, the assignments describing the minimal
      valuations of~$g$ are
      exactly the union of the corresponding assignments for the inputs of~$g$, so we conclude by induction.
    \item For an AND-gate $g$, if $g$ has no inputs, then the only partial trace
      rooted at $g$ is the partial trace $\{g\}$, whose one
      minimal assignment is $\semp$, which sets all variables to~$0$, and $S(g)
      = \ssemp$ is correct. Otherwise, the partial traces rooted at~$g$ are obtained by
      taking $g$ and taking one partial trace rooted at each input of~$g$.
      In
      particular, if there is an input $g'$ such that there is no partial trace rooted
      at~$g'$, then there is no partial trace rooted at~$g$: now as by induction we have
      $S(g') = \emptyset$, so we have indeed set $S(g) = \emptyset$ which is
      correct. Otherwise, remembering
      that an augmented circuit is decomposable (because it has a compatible
      order), as $g$ is an AND-gate, we
      know that the intervals for its inputs must be pairwise disjoint. In
      particular, the partial traces for each input of~$g$ must be on leaves having
      pairwise disjoint intervals. Hence, the minimal valuations of a partial trace
      of~$g$ must be minimal valuations of some partial trace rooted at~$g'$, and
      conversely any choice of minimal valuation for partial traces rooted at the inputs
      of~$g$ can be combined to a minimal valuation of a partial trace rooted at~$g$
      thanks to the fact that the intervals are disjoint. This allows us to
      conclude using the induction
      hypothesis.\qedhere
    \end{itemize}
\end{proof}

\begin{toappendix}
  Observe that Lemma~\ref{lem:botup} allows us to give an equivalent rephrasing
  of the determinism condition in zero-suppressed semantics: a d-DNNF in
  zero-suppressed semantics is a decomposable circuit where, for each OR-gate
  $g$, the captured set $S(g)$ is the \emph{disjoint} union of the captured sets of the
  inputs of~$g$.

We now show the main result of this section:
\end{toappendix}

We can now state our main reduction result for this section: we can rewrite
any d-DNNF to an equivalent d-DNNF in zero-suppressed semantics, by introducing
$\geg 0$-range gates to write explicitly that the variables not tested in a
trace are unconstrained:

\begin{propositionrep}
  \label{prp:reducing}
  Given a d-DNNF circuit $C$ and a compatible order~$<$, we can compute 
  in linear time a monotone $0$-augmented circuit $C^*$ having $<$ as a compatible
  order, such that $C^*$ is a d-DNNF in zero-suppressed semantics and
  such that $S(C^*)$ is exactly the set of satisfying assignments of~$C$.
\end{propositionrep}

\begin{toappendix}
  To prove the result, we will define \emph{complete} circuits, where,
  intuitively, all variables are tested so the semantics makes no difference:

\begin{definition}\label{def:complete}
  Given an augmented circuit or decomposable circuit $C$ with a compatible order~$<$,
  we call $\reach(g)$ the subset of~$C_\var$ of the gates $g'$ that have a
  directed path to $g$, or that are in the interval of a range gate that has a
  directed path to~$g$.
  An AND-gate $g$ of~$C$ is \emph{complete} if we have $\reach(g) = [\min(g),
  \max(g)]$. An OR-gate $g$ of~$C$ is \emph{complete} if, for any input $g'$
  to~$g$, we have $\min(g') = \min(g)$ and $\max(g') = \max(g)$.
  The circuit $C$ is \emph{complete} if every gate of~$C$ is complete, and if
  the interval of the output $g_0$ of~$C$ is the complete set of variables
  $C_\var$.
\end{definition}

  To prove Proposition~\ref{prp:reducing}, the first step is to rewrite the
  input d-DNNF to an equivalent complete augmented d-DNNF (in the standard
  semantics). This will be done using $\geg 0$-gates, which are important to
  make this possible in linear time, and always evaluate to~$1$ so can be added without
  changing the completed function. The second step is to rewrite
  the complete augmented circuit to a monotone augmented circuit in the
  zero-suppressed semantics, whose captured set is the set of satisfying
  assignments of the original circuit.

  Let us first take care of the first step:

\begin{lemmarep}\label{lem:completion}
  Given a d-DNNF circuit $C$ and a compatible order~$<$, we can compute 
  in linear time a complete $0$-augmented circuit $C'$ that has $<$ as a compatible
  order, computes the same function as~$C$, and is a d-DNNF (i.e., all its
  OR-gates are deterministic, with decomposability being implied by the
  compatibility of~$<$).
\end{lemmarep}

The completion procedure of Lemma~\ref{lem:completion} is a standard routine for
many forms of read once branching programs, see
e.g.~\cite[Lemma~6.2.2]{wegener2000branching}.
Note that this lemma is actually the only
place where we use the compatible order~$<$ directly; in all later results, it
  will only be used implicitly, to define the semantics of range gates (if any).

\begin{proof}[Proof of Lemma~\ref{lem:completion}]
 In a first step, it is convenient to make sure that every gate in $C$ has
  fan-in at most two.
  To this end, replace each gate of higher fan-in with a binary tree in the obvious way.
  Note that this can easily be done in linear time and in a way that preserves compatibility with~$<$. 
 
 In a second step we compute for every gate~$g$ in~$C$ the values $\min(g)$ and
  $\max(g)$. Again, this can be done in linear time in a straightforward bottom-up fashion.
  
 Note that the completeness requirement is trivial on gates $g$ with no inputs,
  and it is immediately satisfied on gates $g$ with exactly one input if we
  assume that their one input satisfies the requirement. Hence, it suffices to
  consider gates of fan-in exactly~$2$ in the following. Also note that one
  direction of the completeness requirement is immediate: for any AND- or
  OR-gate $g$, we have $\reach(g) \subseteq [\min(g), \max(g)]$, so only the
  converse implication needs to be proven. We will write $\prec$
  for the \emph{covering relation} of~$<$, i.e., we have $g \prec g'$ if $g < g'$ and
  there is no $g''$ such that $g < g'' < g'$. We will process the gates
  bottom-up to ensure the completeness requirement, assuming by induction that
  it is satisfied on all input gates. As the requirement is immediate for
  variables, NOT-gates and range gates, we first explain the construction for
  AND-gates and second for OR-gates.
 
  First, let $g$ be an AND-gate with inputs $g_1$ and $g_2$. Remember that
  the interval of~$g$ is $[\min(g), \max(g)]$ which is $[\min(g_1), \max(g_2)]$,
  and we know by compatibility of~$<$ that $\max(g_1) < \min(g_2)]$.
  If $\max(g_1)\prec \min(g_2)$, i.e., $\max(g_1)$ is the predecessor of
  $\min(g_2)$, then there is nothing to do for~$g$, as every gate in the
  interval of~$g$ is in that of~$g_1$ or in that of~$g_2$, in which case we
  conclude that it is in $\reach(g_1)$ or in $\reach(g_2)$ by induction
  hypothesis, and conclude.
  
  If $\max(g_1) \not\prec \min(g_2)$, let $g_1'$ and $g_2'$ be such that
  $\max(g_1) \prec g_1'$ and $g_2' \prec \min(g_2)$. Add a fresh child $g'$ to~$g$,
  between $g_1$ and $g_2$, which is a $\geg 0$-range gate with inputs $g_1'$ and
  $g_2'$. It is clear that this does not violate compatibility of~$<$
  with~$C$, because the interval of~$g'$ is $[g'_1, g'_2]$ and we have
  $\max(g_1) < g'_1$ and $g'_2 < \min(g_2)$. Further, it does not change
  the computed function, because $g'$ always evaluate to~$1$ which is neutral
  for AND. Last, it is now clear that $g$ satisfies the condition of
  Definition~\ref{def:complete}, because any gate $g''$ in the interval of~$g$ is now
  either a gate of the interval of~$g_1$, of the interval of~$g_2$, or of the
  interval of~$g'$: we conclude by induction as above in the first two cases,
  and in the third case we conclude because $g''$ is in the interval of the
  range gate $g'$ which has a directed path to~$g$.
  
  Second, let $g$ be an OR-gate with inputs $g_1$ and $g_2$. We replace $g_1$
  with an AND-gate $g'_1$ computing the AND of the following:
  \begin{itemize}
    \item if $\min(g) < \min(g_1)$, letting $g_1^-$ be such that $\min(g) \leq g_1^- \prec
      \min(g_1)$, a $\geg 0$-range gate $(g_1')^-$ whose inputs are $\min(g)$ and $g_1^-$;
    \item $g_1$ itself;
    \item if $\max(g_1) < \max(g)$, letting $g_1^+$ be such that $\max(g_1)
      \prec g_1^+ \leq \max(g)$, a $\geg 0$-range gate $(g_1')^+$ whose inputs are
      $g_1^+$ and $\max(g)$.
  \end{itemize}
  We do the analogous construction for $g_2$. It is clear that this does not
  violate the compatibility of~$<$ with~$C$, because the interval of the new
  AND-gates $g_1'$ and $g_2'$ are exactly the interval of~$g$ by construction,
  and the interval of~$g$ is unchanged.
  It is also clear that these
  transformations do not change the computed function because the $\geg 0$-range
  gates evaluate to 1; further, the transformations do not affect determinism at
  the OR-gate $g$ for the same reason. Last, it is now clear that $g_1'$, $g_2'$
  and $g$ satisfy the conditions of Definition~\ref{def:complete}. Indeed, the
  interval of these three gates is $[\min(g), \max(g)]$. Now, to show the
  condition on $g_1'$ any gate of this
  interval is either in the interval of $g_1$, of $(g_1')^-$, or of $(g_1')^+$,
  we conclude using the induction hypothesis in the first case and immediately
  in the two other cases. To show the condition of~$g$, we use the same proof.
  To show the condition on~$g_2'$, we use the analogous proof with $g_2$,
  $(g_2')^-$, and $(g_2')^+$.
 
 We perform the above constructions for all gates. Last, if the interval of the output gate~$g_0$ does not contain all variables
 of~$C_\var$, we replace it by an AND-gate of~$g_0$ and of ${\geg 0}$-range gates that
 capture the missing variables.
 The resulting
 circuit~$C'$ then computes the same function as~$C$ and is complete. 
 Moreover, $C'$ has a compatible order $<$ and all OR-gates are deterministic,
 so it is a d-DNNF.
 Finally, for every gate $g$ that we manipulated, we only performed a
 construction that can be done in constant time, so the overall time of the
 construction is linear in the size of $C$.
\end{proof}

  We now take care of the second step: given an augmented d-DNNF circuit
  $C'$ which is complete, rewrite it to a monotone augmented circuit $C^*$ which
  is a d-DNNF in zero-suppressed semantics and captures the satisfying
  assignments of~$C'$. We will do so simply by substituting all NOT-gates with
  gates that always evaluate to~$1$:

  \begin{definition}
    Given an (augmented) circuit $C$, its \emph{monotonization} is the monotone
    (augmented) circuit $C^*$ obtained by removing the input wire to each
    NOT-gate of~$C$ and changing the type of these gates to AND-gates (which
    have no children, so they always evaluate to~$1$).
  \end{definition}

  This transformation can clearly be performed in linear time. We now claim that
  it preserves some properties. First, we must make the following trivial
  observation:

  \begin{observation}
    \label{obs:monoorder}
    For any augmented circuit $C$ with compatible order~$<$, its monotonization
    $C^*$ still admits $<$ as a compatible order.
  \end{observation}

  Second, the key observation is the following:

  \begin{lemma}
    \label{lem:monocap}
    For any complete augmented circuit $C'$ capturing Boolean function $\Phi$,
    its monotonization $C^*$ captures $\Phi$ under zero-suppressed
    semantics.
  \end{lemma}

  To prove this lemma, it will be useful to extend the definitions of upward
  trees and traces (Definition~\ref{def:trace}) to augmented circuits which are
  not necessarily monotone; the only difference is that the leaves of a trace
  can now also include NOT-gates.
  We define minimal valuations like in
  Definition~\ref{def:minval}, except that we enforce that negated variables
  must be set to~$0$; the variables tested by a trace also include the negated
  variables. We can now show the important property of complete circuits:

  \begin{observation}
    \label{obs:compltrace}
    Every trace $T$ of a complete circuit $C$ tests all variables of~$C_\var$.
  \end{observation}

  \begin{proof}
    We simply prove by bottom-up induction that any trace rooted at a gate $g$
    of~$C$ tests all variables of the interval of~$g$. This is true of variable
    gates, NOT-gates, and range gates; it is true of OR-gates because it is true
    of all their inputs and they have the same interval as their inputs; it is
    true of AND-gates because the sub-traces on all inputs satisfy the property.
    We conclude thanks to the fact that the interval of the output gate consists
    of all variables of~$C_\var$.
  \end{proof}

  We can now show:

  \begin{proof}[Proof of Lemma~\ref{lem:monocap}]
    We observe the existence of a
    bijection between the traces of~$C^*$ and of~$C'$. Specifically,
    consider the mapping
    $f$ from the traces of~$C'$ to traces of~$C^*$ obtained by replacing each
    leaf which is a NOT-gate of~$C'$ by the corresponding AND-gate with no
    inputs in~$C^*$. It is clear that any upward tree in
    the image of this transformation is a trace of~$C^*$. Conversely, we can map
    the traces of~$C^*$ to traces of~$C'$ by replacing the fresh AND-gates with
    no inputs by the
    corresponding NOT-gate, and again the upward trees in the image of this
    transformation are indeed traces of~$C'$. As this defines an inverse function
    for~$f$, it is clear that~$f$ is a bijection. Further, as $C'$ is complete,
    it is clear that the variables which are not tested by~$f(T')$ are precisely
    the variables whose negation is a leaf of~$'T$

    For the forward direction, consider a satisfying valuation $\nu'$ of~$C'$. By
    Observation~\ref{obs:trace}, there is a trace $T'$ of~$C'$ of which $\nu'$ is a
    satisfying valuation. As $C'$ is complete, by
    Observation~\ref{obs:compltrace}, $T'$ tests all variables of~$C'_\var$.
    Hence, the satisfying valuation $\nu'$ of~$T'$ is actually a minimal
    satisfying valuation of~$T$ (there are no variables implicitly set to~$0$).
    Now, clearly $\nu'$ is
    also a satisfying valuation of~$f(T')$. The converse is shown in the same
    way by considering a satisfying valuation $\nu^*$ of~$C^*$, a witnessing trace $T^*$
    of~$C^*$ of which it is a minimal valuation, and observing that $\nu^*$ is a
    satisfying valuation of the trace $f^{-1}(T^*)$ of~$C'$.
  \end{proof}

  The last claim to show is the preservation of d-DNNF through 
  monotonization:

  \begin{lemma}
    \label{lem:monodnnf}
    For any complete augmented circuit $C'$ which is a d-DNNF in the standard semantics,
    its monotonization $C^*$ is a d-DNNF in zero-suppressed semantics (in the
    sense of Definition~\ref{def:zssddnnf}).
  \end{lemma}

  \begin{proof}
    As in the proof of Lemma~\ref{lem:monocap}, we know that the captured set of each gate $g$ of~$C^*$ describes the satisfying assignments of~$g$ in~$C'$ in the standard semantics. Hence, any violation of determinism in zero-suppressed semantics on~$C^*$ witnesses a violation of determinism in the standard semantics in~$C'$.
  \end{proof}

  We are now ready to prove our main result for this section:

  \begin{proof}[Proof of Proposition~\ref{prp:reducing}]
  We first apply to~$C$ the construction of Lemma~\ref{lem:completion}
  to get in linear time a complete 0-augmented circuit $C'$ which is a d-DNNF in
    the standard sense,
    has compatible order~$<$, and computes the same function
    $\Phi$ as $C$. Now, we construct in linear time the monotonization $C^*$
    of~$C'$, which is a monotone 0-augmented circuit. 
    By Observation~\ref{obs:monoorder}, $C^*$ still admits~$<$ as
    compatible order. By Lemma~\ref{lem:monodnnf}, $C^*$ is a d-DNNF in the
    zero-suppressed semantics. By Lemma~\ref{lem:monocap}, $C^*$ captures $\Phi$
    under zero-suppressed semantics, which concludes the proof.
\end{proof}

We conclude the section by a remark on zero-suppressed semantics: this
semantics, as well as the relevant definitions, can be defined not only for
augmented circuits, but more generally for decomposable (non-augmented)
circuits:

\begin{remark}
  \label{rmk:zerosuppressed}
  The definition of upwards trees and traces (Definition~\ref{def:trace})
  extend to monotone (non-augmented) circuits which are decomposable, even if
  they do not have a compatible order.
  Observation~\ref{obs:trace} extends to them, and the definition of minimal
  valuations (Definition~\ref{def:minval}) also does.
  Further, the definition of zero-suppressed
  semantics extends (Definition~\ref{def:zss}), as does the definition of
  captured sets in Definition~\ref{def:capset}.
  Last, the alternative characterization of the captured sets in
  Lemma~\ref{lem:botup} also extends.
\end{remark}

Thanks to this remark, we will be able to talk about zero-suppressed
semantics and captured sets in decomposable monotone circuits with no range
gates, even if they do not have a compatible order. This will be especially
useful in Section~\ref{sec:applications}, where we directly produce decomposable
circuits in zero-suppressed semantics.
\end{toappendix}

\section{Reducing to Normal Form Circuits}
\label{sec:normal}
In this section, given Proposition~\ref{prp:reducing}, we work on a
monotone $0$-augmented d-DNNF circuit~$C$ in zero-suppressed semantics, with a compatible order~$<$ to define
the semantics of range gates.
We present our next two preprocessing steps for the enumeration of
the assignments $S(C)$ of~$C$:
restricting our attention to valuations of the right Hamming weight (for
Theorem~\ref{thm:constant} only), and bringing $C$ to a normal form that makes enumeration easier.

\begin{toappendix}
  Throughout this appendix, two monotone circuits $C$ and $C'$ are called
  \emph{equivalent} (in zero-suppressed semantics) if $S(C) = S(C')$.
\end{toappendix}

\subparagraph*{Homogenization.}

\begin{toappendix}
\subsection{Reduction to Arity-Two}
 We show that following easy general-purpose result which will be useful in several proofs.

\begin{lemma}
  \label{lem:arity2}
  For any (augmented) monotone circuit $C$,
  we can compute in linear time an equivalent arity-two (augmented) monotone circuit $C'$.
  Further, if $C$ is a d-DNNF in zero-suppressed semantics (resp., if it is
  $\emptyset$-pruned, if it is $\semp$-pruned, if it has some order $<$ as a
  compatible order, if it is monotone), then the same is true of~$C'$.
\end{lemma}

\begin{proof}
  We use the standard construction of adding intermediate AND-
  and OR-gates, leveraging the associativity of the Boolean $\land$ and $\lor$
  operations. It is clear
that
  none of the existing or additional gates is unsatisfiable or 0-valid if
  none of the original gates are, so the process preserves being
  $\emptyset$-pruned and $\semp$-pruned. The process
  adds no NOT-gates so it clearly preserves monotonicity. Any violation
of the d-DNNF condition on the rewritten circuit witnesses a violation on the
original circuit. Last, it is clear that any compatible order is still
  compatible with the result of the rewriting, as any violation of compatibility would
  witness a violation in the original circuit.
\end{proof}
\end{toappendix}

\begin{toappendix}
\subsection{Homogenization}
  \label{apx:homogenization}
\end{toappendix}
Our input augmented circuit $C$ in zero-suppressed semantics may have satisfying
assignments of arbitrary Hamming weight. When proving Theorem~\ref{thm:main}, this is
intended, and the construction that we are about to describe is not necessary. However, 
when proving
Theorem~\ref{thm:constant} about enumerating valuations of constant weight, we
need to restrict our attention to such valuations, to ensure constant delay. We
do so using the following homogenization result,
adapted from the technique of Strassen~\cite{Strassen73}:

\begin{proposition}
  \label{prp:homogenize}
  Given $k\in\NN$ and a monotone augmented d-DNNF circuit $C$ in zero-suppressed semantics
  with compatible order $<$, we can construct in
  time $O(k^2 \cdot \card{C})$ a monotone augmented d-DNNF circuit $C'$
  in zero-suppressed semantics 
  with compatible order $<$ such that
  $S(C') = \{t \in S(C) \mid \card{t} \leq k\}$.
\end{proposition}

\begin{proofsketch}
  We create $k+2$ copies of each gate $g$, with each copy capturing the
  assignments of a specific weight from~$0$ to~$k$ inclusive (or, for the $k+2$-th copy, the
  assignments with weight $>k$). In particular, for $\geg 0$-gates $g$, for $0 \leq
  i \leq k$, we use an $\eqg i$-gate for the copy of~$g$ capturing weight $i$. We then
  re-wire the circuit so that weights are correctly preserved.
\end{proofsketch}

Note that this
is the only place where our preprocessing depends on~$k$:
in particular, for constant $k$, the construction is linear-time. This result allows us
to assume in the sequel that the set of assignments of the circuit in 
zero-suppressed semantics contains precisely the valuations that
we are interested in, i.e., those that have suitable Hamming weight.

\begin{toappendix}

  The proof will use the following definition:
\begin{definition}
  \label{def:homogenization}
  For $k \in \NN$,
  a monotone augmented circuit $C$ is \emph{$k$-homogenized} if its gate set $G$
  is partitioned as $G = G^{=0} \cup G^{=1} \cup \cdots \cup G^{=k} \cup
  G^{>k}$, all these unions being disjoint,
  such that for $i\in [0,k]$ the $G^{=i}$ and
$G^{>k}$ satisfy the following properties:
\begin{itemize}
\item For every $g \in G^{=i}$, for every $t \in S(g)$, we have $\card{t} = i$, and 
\item For every $g \in G^{>k}$, for every $t \in S(g)$, we have $\card{t} > k$.
\end{itemize}
  Note that $S(g) = \emptyset$ is allowed in both cases. Note that in particular
  variable gates are all in $G^{=1}$ if $k \geq 1$, and they are all in $G^{>0}$
  if $k = 0$.
  
  $C$ is a called a \emph{$k$-homogenization} of an augmented circuit $C'$, if for every gate $g$ of $C'$
  \begin{itemize}
   \item For every $i\in [0,k]$, $C$ contains a gate $g^{=i}$ such that $S(g^{=i}) = \{t\in S(C')\mid \card{t} = i\}$, and 
   \item $C$ contains a gate $g^{>k}$ such that $S(g^{>k}) = \{t\in S(C')\mid \card{t} >k\}$.
  \end{itemize}
\end{definition}

  We will now prove the following strengthening of Proposition~\ref{prp:homogenize}:

\begin{proposition}
  \label{prp:homogenize2}
  For every $l \in \NN$, given a monotone $l$-augmented $C$ with compatible
  order~$<$ and $k \in \NN$, we can construct in
  time $O((k+1)^2 \cdot \card{C})$ a monotone $\max(l, k+1)$-augmented circuit
  $C'$ with compatible order~$<$ 
  such that $C'$ is a $k$-homogenization of $C$. Further, if $C$ is a d-DNNF in
  zero-suppressed semantics then $C'$ also is.
\end{proposition}

  It is clear that this result implies Proposition~\ref{prp:homogenize}.
  Indeed,
  to impose the desired semantics on the
  resulting circuit $C'$, one 
  can simply add a fresh OR-gate as the output gate of
  the $k$-homogenization $C'$ as the OR of the $G^{=i}$ for $0 \leq i \leq k$.
  It is then clear that the circuit satisfies the required properties, and we
  easily check determinism on the fresh output gate in the sense of
  Definition~\ref{def:zssddnnf} thanks to the fact that the valuations are
  disjoint (they have different weights).

  We now prove Proposition~\ref{prp:homogenize2}:

  \begin{proof}[Proof of Proposition~\ref{prp:homogenize2}]
  The construction essentially follows the classical homogenization technique
    introduced by Strassen~\cite{Strassen73}, and will not use the input
    compatible order except for the definition of range gates. We first rewrite the input circuit
    $C$ in linear time using Lemma~\ref{lem:arity2} to ensure that it is
    arity-two. We will write $\ins(g)$ to denote the inputs of a gate~$g$.
  
  Now, for every gate $g$ of $C$, for every $i\in \NN$, we define the sets $S^{=i}(g)\colonequals \{t\in S(G)\mid \card{t} = i\}$ and $S^{>i}(g) \colonequals \{t\in S(G)\mid \card{t} > i\}$.
  We create the $k$-homogenized circuit $C'$ by associating,
  to each gate $g$ of $C$, $k+2$ gates $g^{=0}, \ldots, g^{=k}$ and one gate
    $g^{>k}$ in $C'$ such that for $i\in [0,k]$ we have $S(g^{=i}) = S^{=i}(g)$
    and $S(g^{>k}) = S^{>k}(g)$. To ensure this, we proceed iteratively as follows:
  \begin{itemize}
    \item If $g$ is a variable, for all $i\in [0, k]\setminus \{1\}$, the gate $g^{=i}$ is an OR-gate with no inputs (so
      $S(g^{=0}) = \emptyset$). If $k>0$, $g^{=1}$ is a variable identified
      to~$g$ and $g^{>k}$ is an OR-gate with no inputs. Otherwise, $g^{>0}$ is a
      variable identified to~$g$.
    \item If $g$ is a $\eqg k'$-range gate, 
      then if $k' \leq k$ we set $g^{=k'}$ to be an $\eqg k'$-range gate with
      the same inputs as~$g$, otherwise we have $k < k'$ and set $g^{>k}$ to be an $\eqg
      k'$-range gate with the same inputs as~$g$. All other gates are OR-gates with no
      inputs.
    \item If $g$ is a $\geg k'$-range gate, for all $i\in [0, k'-1]$, the gate $g^{=i}$ is an OR-gate with no inputs. For all $i\in [k', k]$, the gate $g^{=i}$ is a ${\eqg i}$-range gate with the same inputs as $g$. Finally, $g^{>k}$ is a ${\ge (k+1)}$-range gate identified with the same inputs as $g$.
    \item If $g$ is an OR-gate, for every $i\in [0,k]$, the gate $g^{=i}$ is an OR-gate whose inputs are 
$\{(g')^{=i} \mid g' \in \ins(g)\}$, and $g^{>k}$ is an OR-gate whose inputs
      are $\{(g')^{>k} \mid g' \in \ins(g)\}$;
    \item If $g$ is an AND-gate, for every $i\in [0,k]$, the gates $g^{=i}$ and $g^{>k}$ are defined as follows: (remember that $C$ is
      arity-two so the following cases are exhaustive):
      \begin{itemize}
        \item if $\card{\ins(g)} = 0$, then $g^{=i}$ and $g^{>k}$ are AND-gates with $\ins(g^{=i}) \colonequals \emptyset$ and 
          $\ins(g^{>k}) \colonequals \emptyset$;
        \item if $\card{\ins(g)} = 1$, writing $\ins(g) = \{g'\}$, then $g^{=i}$ and $g^{>k}$ are AND-gates with $\ins(g^{=i}) \colonequals \{(g')^{=i}\}$ and 
          $\ins(g^{=i}) \colonequals \{(g')^{>i}\}$;
        \item if $\card{\ins(g)} = 2$, writing $\ins(g) = \{g'_1, g'_2\}$, then for $i\in [0,k]$, the gate $g^{=i}$ is an OR-gate with inputs $g^{=i}_0, \ldots, g^{=i}_i$ where each $g^{=i}_j$ is an AND-gate with inputs $(g'_1)^{=j}$ and $(g'_2)^{=(i-j)}$. Moreover, $g^{>k}$ is an OR-gate with inputs $g^{>k}_{i,j}$ for $i,j\in [0,k]$ and $g^{>k,1}_i$ and $g^{>k,2}_i$ for $i\in [0,k]$ which are defined as follows:
        \begin{itemize}
        \item if $i+j>k$, then $g^{>k}_{i,j}$ is an AND-gate with inputs $(g'_1)^{=i}$ and $(g'_2)^{=j}$, 
        \item if $i+j\le k$, then $g^{>k}_{i,j}$ is an OR-gate with no inputs,
        \item $g^{>k,1}_i$ is an AND-gate with inputs $(g'_1)^{>k}$ and $(g'_2)^{=i}$
        \item $g^{>k,2}_i$ is an AND-gate with inputs $(g'_1)^{=i}$ and $(g'_2)^{>k}$
        \end{itemize}
      \end{itemize}
  \end{itemize}

  By straightforward but somewhat cumbersome induction, it is easy to see that the gates $g^{=0}, \ldots, g^{=k}$ and $g^{>k}$ indeed compute the sets $S^{=0}(g), \ldots, S^{=k}(g)$ and $S^{>k}(g)$, respectively. Thus, $C'$ is a $k$-homogenization of $C$ as claimed. Moreover, the construction for every gate $g$ in $C$ can be performed in time $O(k^2)$, so the overall runtime of the algorithm is $O(k^2 \cdot |C|)$.
  
    It is clear that the resulting circuit $C'$ is monotone. If $C$ has a
    compatible order $<$, we first show that
    $<$ is a compatible order for~$C'$, so that $C'$ is $l$-augmented. Indeed, all AND-gates $g^{\land}$ with more than one input in $C'$ derive from the construction for an AND-gate $g$ with inputs $g_1'$ and $g_2'$ in $C$. But then in $g^{\land}$ always has inputs $(g_1')^{\bowtie_1 i_1}$ and $(g_2')^{\bowtie_2 i_2}$ for ${\bowtie_1}, {\bowtie_2} \in \{=,>\}$ and $i_1,i_2\in [0, k]$. Note that, by construction of~$C'$, the variables having a directed path to $(g'_1)^{\bowtie_1 i_1}$ and $(g'_2)^{\bowtie_2 i_2}$ in~$C'$ are subsets of those having a directed path to $g'_1$ and $g'_2$ respectively in~$C$. Hence, as $<$ is compatible with~$C$, from the condition on $g$ with inputs $g'_1$ and $g'_2$ in~$C$, we deduce that the condition is satisfied for $g^\land$ in~$C'$.
    Hence, $<$ is a compatible order for $C'$.
    We last observe that the range gates that we create are not labeled with
    integer values that are larger than those of the original circuit or than
    $k+1$, so $C'$ is indeed
    a monotone $\max(l,k+1)$-augmented circuit. 
   
    We last show that if $C$ is a d-DNNF in zero-suppressed semantics then
    so is $C'$. To do so, we show that the OR-gates of~$C'$ are deterministic in the sense of Definition~\ref{def:zssddnnf}.
  To this end, consider an OR-gate $g^\lor$ in $C'$ with at least two inputs. Only two cases can occur:
  \begin{itemize}
    \item \emph{$g^\lor$ was introduced in the construction for an OR-gate $g$ in $C$.} In this case, the inputs of $g^\lor$ are $\{(g')^{\bowtie i} \mid g' \in \ins(g)\}$ for ${\bowtie} \in \{=, >\}$ and $i\in [0,k]$. As before $S((g')^{\bowtie i})\subseteq S(g')$ for every input $g'$ of $g$ in $C$. Since $g$ is deterministic, we know that, for all gates $g', g''\in \ins(g)$ with $g'\ne g''$, the sets $S(g')$ and $S(g'')$  are disjoint. It follows that for every $(g')^{\bowtie i}, (g'')^{\bowtie i}\in \ins(g^\lor)$ with $g'\ne g''$, the sets $S((g')^{\bowtie i})$ and $S((g'')^{\bowtie i})$ are disjoint. Thus $g^\lor$ is deterministic.
    \item \emph{$g^\lor$ was introduced in the construction for an AND-gate $g$ in $C$.} In this case, the inputs of $g^\lor$ are OR-gates with no inputs (which can never cause a violation of determinism because they capture the empty set), and AND-gates with inputs of the form $(g_1')^{\bowtie_1 i_1}$ and $(g_2')^{\bowtie_2 i_2}$ where ${\bowtie_1}, {\bowtie_2} \in\{=,>\}$ and $i_1, i_2\in [0,k]$. Let now $g$ and $g'$ be two inputs of $g^\lor$ where $g$ has the inputs $(g_1')^{\bowtie_1 i_1}$ and $(g_2')^{\bowtie_2 i_2}$ and $g'$ has the inputs $(g_1')^{\bowtie_1' i_1'}$ and $(g_2')^{\bowtie_2' i_2'}$. By inspection of the construction, we see that we cannot have ${\bowtie_1} = {\bowtie_1'}$ and ${\bowtie_2} = {\bowtie_2'}$ and $i_1 = i_1'$ and $i_2 = i_2'$ at the same time, that is, one of these equalities must be false. But then, depending on whether the false equality is on $\bowtie_1$ or $i_1$, or on $\bowtie_2$ or $i_2$, we have $S((g_1')^{\bowtie_1 i_1})\cap S((g_1')^{\bowtie_1' i_1'})=\emptyset$ or $S((g_2')^{\bowtie_2 i_2})\cap S((g_2')^{\bowtie_2' i_2'})=\emptyset$. Thus, in both cases $S(g) \cap S(g')=\emptyset$ by Lemma~\ref{lem:botup} and thus $g^\lor$ is deterministic.
  \end{itemize}
  Since in both cases $g^\lor$ is deterministic, it follows, as claimed,  that $C'$ is a d-DNNF in zero-suppressed semantics.
\end{proof}

\end{toappendix}

\subparagraph*{Normal form.} 
\begin{toappendix}
  \subsection{Reduction to Normal Form}
\end{toappendix}

Now that we have focused on the interesting
valuations of our circuit $C$, we can bring it to our desired normal form:

\begin{definition}
  \label{def:normalform}
  A \emph{normal} circuit $C$ is a monotone 
  augmented
  circuit such that:
  \begin{itemize}
    \item $C$ is \emph{arity-two}, i.e., each gate has fan-in at most two.
    \item $C$ is \emph{$\emptyset$-pruned}, i.e., no gate $g$ is unsatisfiable
      (i.e., each gate has some minimal valuation).
    \item $C$ is \emph{$\semp$-pruned}, i.e., no gate $g$ is 0-valid
      (i.e., the valuation that sets all variables to~$0$ is not a minimal
      valuation for any gate).
    \item $C$ is \emph{collapsed}, i.e., it has no AND-gate with fan-in 1.
    \item $C$ is \emph{discriminative}, i.e., for every OR-gate $g$ with an
      input that is not an OR-gate (we call $g$ an \emph{exit}), $g$ has
      fan-in~$1$, fan-out~$1$, and
      the one gate with $g$ as input is an OR-gate.
  \end{itemize}
  $C$ is a
  \emph{normal d-DNNF} if it is additionally a d-DNNF in the
  zero-suppressed semantics.
\end{definition}

The pruned requirements slightly weaken the expressiveness of normal circuits~$C$, because
they forbid that $S(C) = \emptyset$ or that $\semp \in S(C)$. These cases
will be easy to handle separately. The main result of this section is then the
following:

\begin{propositionrep}
  \label{prp:normalize}
  Given a monotone augmented d-DNNF circuit $C$ in
  zero-suppressed semantics with compatible order~$<$ and with $S(C)
  \neq \emptyset$ and $S(C) \neq \ssemp$, we can
  build in $O(\card{C})$ a
  normal d-DNNF $C'$,
  with $<$ as a compatible order, such that $S(C') = S(C) \backslash \ssemp$.
\end{propositionrep}

\begin{proofsketch}
  We reuse the construction of Proposition~\ref{prp:homogenize} with $k=1$ to split the
  gates so that they are not 0-valid, we eliminate bottom-up the unsatisfiable
  gates,
  we make $C$ arity-two in a straightforward way, we collapse all AND-gates
  with fan-in~1, and we make $C$ discriminative by inserting new OR-gates (i.e., the
  exits) on all wires from non-OR-gates to OR-gates.
\end{proofsketch}

This result allows us to assume in the sequel that we are working with
normal d-DNNFs.

\begin{toappendix}
In the rest of this section, we prove Proposition~\ref{prp:normalize} in
multiple steps to ensure each condition.
  
The first step is to make the circuit $\emptyset$-pruned and $\semp$-pruned. To
do so, it will be convenient to reuse our homogenization process
  (Proposition~\ref{prp:homogenize2} of Appendix~\ref{apx:homogenization}) and
  our notion of homogenization of circuits
  (Definition~\ref{def:homogenization}). We first show how to make circuits
  $\emptyset$-pruned while preserving being a d-DNNF and being homogenized:

\begin{lemma}
  \label{lem:empprune}
  For any $l\in \NN$ and monotone $l$-augmented circuit $C$ with compatible order~$<$ such that $S(C) \neq \emptyset$, we can
  compute in linear time an equivalent $\emptyset$-pruned monotone $l$-augmented circuit
  $C'$ with compatible order~$<$. Further,
  if $C$ is a d-DNNF in zero-suppressed semantics, then so is~$C'$, and if $C$
  is a $k$-homogenized for some $k \in \NN$ then so is~$C'$.
\end{lemma}

\begin{proof}
  We first compute which gates $g$ of~$C$ are unsatisfiable. It is easily
  seen that these gates are the following, which can be computed in linear time
  by processing $C$ bottom-up:
  \begin{itemize}
    \item range gates labeled with $\bowtie i$ whose interval contains less
      than~$i$ variables (which we call \emph{unsatisfiable} range gates);
    \item AND-gates with no inputs gates;
    \item AND-gates where one input gate is unsatisfiable;
    \item OR-gates where all input gates is unsatisfiable.
  \end{itemize}
  We define the circuit $C'$ as $C$ where we remove all unsatisfiable gates and all wires leading out of these gates.  This can clearly be
  computed in linear-time. Further, $C'$ has a output gate
  (namely, the same as~$C$), because as $S(C) \neq \emptyset$, we know that we
  have not removed the output gate of~$C$.

  We will show that, for every gate $g$ of $C'$, the set $S(g)$ in~$C'$ is the
  same as $S(g)$ in~$C$. This claim implies in particular that $C'$ is equivalent
  to~$C$ (when applying it to the output gate), and it shows that
  $C'$ is $\emptyset$-pruned: if some gate $g$ is unsatisfiable in~$C'$, then $g$ is also unsatisfiable in~$C$,
  so $g$ should have been removed in~$C'$. To see why the
  claim is true, observe that the only wires from a removed gate to a
  non-removed gate, i.e., from an unsatisfiable gate $g$ to a satisfiable gate $g'$,
  must be such that $g'$ is an OR-gate (otherwise it
  would be unsatisfiable too), and clearly removing the wire from $g$ to $g'$
  does not change the set captured by~$g'$.
   
  Assuming now that $C$ is $k$-homogenized for some $k
  \in \NN$, we can see from
  the previous claim that the same is
  true of~$C'$. Indeed, we can suitably partition the gates
  of~$C'$ using the same partition as the one used for~$C$.

  Last, to see that $C'$ admits $<$ as a compatible order, and that $C'$ is a d-DNNF
  in zero-suppressed semantics if $C$ is, observe that we construct $C'$
  from~$C$ by removing gates and removing input wires to the remaining gates, so
  this cannot introduce violations of the compatibility of~$<$ or the
  determinism of OR-gates.
\end{proof}

We now show how to make the circuit $\emptyset$-pruned and $\semp$-pruned, using
the previous process and the homogenization process
(Proposition~\ref{prp:homogenize2} of Appendix~\ref{apx:homogenization}).

\begin{lemma}
  \label{lem:pruned}
  For any $l \in \NN$, for any monotone $l$-augmented circuit $C$ with
  compatible order~$<$ such that
  $S(C) \neq \emptyset$ and $S(C) \neq \ssemp$, we can compute in linear time a monotone $\max(l,
  1)$-augmented circuit $C'$ with compatible order~$<$ such that $S(C') = S(C) \backslash \ssemp$ and $C'$
  is $\emptyset$-pruned and $\semp$-pruned.
  Further, if $C$ is a d-DNNF in zero-suppressed semantics, then so
  is~$C$.
\end{lemma}

\begin{proof}
  We use Proposition~\ref{prp:homogenize2} for $k = 1$ to compute in linear time
  in~$C$ a monotone $\max(l, 1)$-augmented circuit $C_{\text{homog}}$ which is a $0$-homogenization
  of~$C$, which admits $<$ as a compatible order, and which is a d-DNNF
  according to zero-suppressed semantics iff $C$
  is. Recalling now Definition~\ref{def:homogenization}, we know that
  $C_{\text{homog}}$ has a
  gate $g^{>0}$ such that $S(g^{>0}) = \{t \in S(C) \mid \card{t} > 0$, so
  choosing $g^{>0}$ as the output gate of~$C_{\text{homog}}$ we have indeed that
  $S(C_{\text{homog}}) = S(C) \backslash \ssemp$; in particular
  $S(C_{\text{homog}}) \neq \emptyset$.

  We now apply Lemma~\ref{lem:empprune} to~$C_{\text{homog}}$, to obtain an
  equivalent $\emptyset$-pruned monotone $\max(l,1)$-augmented circuit
  $C_{\emptyset}$
  which is $\emptyset$-pruned, which is is still $0$-homogenized, and which is a
  d-DNNF in zero-suppressed semantics if~$C$ is.

  We last rewrite $C_{\emptyset}$ to an equivalent circuit $C'$. Recall that, as
  $C'$ is $0$-homogenized, its gate set is partitioned in
  $G^{=0}$ and $G^{>0}$, such that all gates of $G^{=0}$ capture 
  $\ssemp$ (they cannot capture $\emptyset$ as $C_\emptyset$ is
  $\emptyset$-pruned), and no gates of the latter capture a set containing
  $\semp$ (and $G^{>0}$ contains in particular the output gate).
  We define our final circuit $C'$ from $C_\emptyset$ by removing all gates of~$G^{=0}$
  and all wires leading out of them. Note that, by definition of homogenized
  circuits, we do not remove variable gates or the output gate.
  The construction of~$C'$ is clearly in linear-time, and $C'$ is compatible
  with $<$ and is a d-DNNF in
  zero-suppressed semantics if $C'$ is, because $C'$ is constructed from
  $C_\emptyset$ by removing gates and input to remaining gates, which cannot
  introduce violations of these requirements.

  We now show that for every gate $g$ of $C'$, its captured set $S(g)$ in~$C'$
  is the same as $S(g)$ in~$C_\emptyset$.
  This claim implies in particular that $C'$ is equivalent
  to~$C_\emptyset$ (when applying it to the output gate),
  that $C'$ is $\emptyset$-pruned (any violation of this in~$C'$ implies a
  violation of the fact that $C_\emptyset$ is $\emptyset$-pruned),
  and it shows that $C'$ is $\semp$-pruned: if $\semp \in S(g)$ in~$C'$ for some
  gate~$g$, then the same holds of~$g$ in~$C_\emptyset$, so $g \in G^{=0}$ and
  $g$ should have been removed in~$C'$.
  To see why the claim is true, observe that when there is a wire from a
  gate $g$ in $C'$ to a gate $g'$ in~$C'$, and $g$ is removed and $g'$ is
  not, i.e., $g \in G^{=0}$ and $g' \in G^{>0}$,
  then we have $S(g) \subseteq \ssemp$, and then as $C'$ is $\emptyset$-pruned we must
  have $S(g) = \ssemp$. Hence, we know that $g'$ is an AND-gate, because if it were an
  OR-gate we would have $\semp \in S(g')$ contradicting $g' \in G^{>0}$. Now as
  $\semp$ is neutral for~$\times$, we do not change the semantics by removing the
  wire.

  Hence, $C'$ is a monotone $\max(l,1)$-augmented circuit with compatible order~$<$ that
  is $\emptyset$-pruned and $\semp$-pruned, we have $S(C') = S(C_\emptyset) =
  C_{\text{homog}} = S(C) \backslash \ssemp$, and if $C$ is a d-DNNF in the
  zero-suppressed semantics then $C_{\text{homog}}$, $C_\emptyset$, and thus
  $C'$, also are. This concludes the proof.
\end{proof}

The third step is to ensure that the circuit is collapsed and discriminative,
but this is completely straightforward:

\begin{lemma}
  \label{lem:collapsed}
  For any $l\in\NN$ and $l$-augmented circuit $C$ with compatible order~$<$,
  we can compute in linear time an $l$-augmented circuit $C'$ with compatible
  order~$<$ which is collapsed and discriminative. Further, if $C$ is arity-two
  (resp., is $\emptyset$-pruned, is $\semp$-pruned, is a d-DNNF in the
  zero-suppressed semantics, is monotone), then so is~$C'$.
\end{lemma}

\begin{proof}
  Simply merge all AND-gates with one input with their one input to make
  the circuit collapsed. This is clearly linear-time, and does not affect
  compatibility with~$<$, OR-determinism, being 
  arity-two, being $\emptyset$-pruned, or being $\semp$-pruned.
  
  Then, for every wire $(g, g')$ where the input gate
  $g$ is not an OR-gate but the output gate~$g'$ is an OR-gate, rewrite
  the wire by inserting an intermediate OR-gate, i.e., we create a fresh OR-gate
  $g''$ (the exit), and replace the wire $(g, g')$ by $(g, g'')$ and $(g'',
  g')$. This is clearly linear-time, ensures that the output is discriminative,
  and it does not affect any of the
  requirements.
\end{proof}

Using these results, we can conclude the proof of
Proposition~\ref{prp:normalize}:

\begin{proof}[Proof of Proposition~\ref{prp:normalize}]
  We first apply Lemma~\ref{lem:pruned} to make the circuit $\emptyset$-pruned
  and $\semp$-pruned. We then apply Lemma~\ref{lem:arity2} to make it arity-two.
  We last apply Lemma~\ref{lem:collapsed} to make it collapsed and
  discriminative.
\end{proof}

  As in Appendix~\ref{apx:reducing}, we will need in
  Section~\ref{sec:applications} to apply the process described
  in this section to non-augmented circuits that have no compatible order but
  are decomposable. The claim is as follows, and it is straightforward to
  verify, because all transformations described in this section do not introduce
  range gates if their input does not contain range gates, and do not depend on
  the compatible order except to define the semantics of range gates and to
  ensure decomposability.

  \begin{remark}
    \label{rmk:rangegates}
    If the input circuit $C$ to Propositions~\ref{prp:homogenize}
    or~\ref{prp:homogenize2} has no range
    gates and no compatible order, but is decomposable, then the construction
    still works, and the output still does not have range gates and is still
    decomposable. The same is true of Proposition~\ref{prp:normalize}.
  \end{remark}

\end{toappendix}

\section{Indexing OR-Components}
\label{sec:multitrees}
\begin{toappendix}
  \label{apx:multitrees}
\end{toappendix}

This section presents the last step of our preprocessing. Remember that we now work
with a normal d-DNNF, and we want to enumerate its set of assignments.
Intuitively, this last preprocessing will help us to
enumerate the choices that can be made at OR-gates. Formally, we
will work on the \emph{OR-components} of our circuit:

\begin{definition}
  \label{def:orcomponent}
  The \emph{OR-component} $K$ of an OR-gate $g$ in a normal circuit $C$ is the
  set of OR-gates that can be reached from~$g$ by going only through
  OR-gates, following wires in either direction. We abuse notation and also see
  $K$ as a DAG, whose
  vertices are the gates of~$K$, and whose edges are the wires between them.
\end{definition}

Recall from Definition~\ref{def:normalform} that, as $C$ is discriminative, all
gates of an \mbox{OR-component}~$K$ with no inputs in $K$ must be exits; we call
them the \emph{exits} of~$K$. For a gate $g$ in~$K$, the
\emph{exits} of~$g$ are the gates of~$K$ that have a directed path to~$g$ in~$K$;
intuitively, they are the ``possible choices'' for a partial trace rooted
at~$g$.
Our goal is to preprocess each OR-component of~$C$ to be able to enumerate
efficiently the exits of all OR-gates of~$C$.
This enumeration task is tricky, however:
exploring~$K$ naively when enumerating would take time dependent of~$C$, but
materializing a reachability index would take quadratic preprocessing time.
Thus, we design an efficient indexing scheme, using the fact that OR-components are \emph{multitrees}:

\begin{definition}
  A DAG $G$ is a \emph{multitree} if it has no pair $n \neq n'$ of vertices such
  that there are two different
  directed paths from~$n$ to~$n'$. In particular, forests are multitrees, and so
  are polytrees (DAGs with no undirected cycles).
\end{definition}

\begin{lemmarep}
  \label{lem:multitrees}
  For any normal d-DNNF $C$, each
  OR-component of~$C$ is a multitree.
\end{lemmarep}

\begin{proof}
  Assume by contradiction that an OR-component is not a multitree, so it has
  two gates $g$ and $g'$ such that there are two different directed paths
  $\pi_1$ and $\pi_2$ from
  $g$ to $g'$. As $\pi_1$ and $\pi_2$ are two different paths to the same
  gate~$g'$, there must be a gate $g''$ with inputs $g_1'' \neq g_2''$ 
  such that $\pi_1$ goes through $g_1''$ and $g''$, and $\pi_2$ goes through
  $g_2''$ and $g''$.  
  As $C$ is
  $\emptyset$-pruned, $S(g')$ is non-empty, so let $t \in S(g')$. As $C$ is
  $\semp$-pruned, $t$ is non-empty. The
  directed paths $\pi_1$ and $\pi_2$ witness that $t \in S(g_1'')$ and $t \in
  S(g_2'')$, and this
  violates the determinism condition on the OR-gate~$g''$.
\end{proof}

We can then prepare the enumeration of exits of gates in OR-components,
by designing an efficient and generic indexing scheme
on \emph{multitrees} (see Appendix~\ref{apx:multitrees}). We deduce:

\begin{theoremrep}\label{thm:orindex}
  Given a normal d-DNNF $C$, we can compute in $O(\card{C})$ a structure
  called \emph{\mbox{OR-index}} allowing
  us to do the following: given an OR-gate $g$ of~$C$, enumerate the exits
  of~$g$ in its OR-component~$K$,
  with constant delay and memory usage $O(\log \card{K})$.
\end{theoremrep}

\begin{toappendix}
  In order to prove Theorem~\ref{thm:orindex}, thanks to 
  Lemma~\ref{lem:multitrees}, it suffices to show the following general result
  on multitrees, where a
  \emph{leaf} of a multitree~$T$ is a vertex $n$ with no edge to a vertex
  of~$T$:

\begin{theorem}
  \label{thm:multitreeidx}
  Given a multitree $T$, we can compute in linear time a data structure allowing
  us to perform the following: given $n \in T$, enumerate the
  leaves of~$T$ that are reachable from~$n$ with constant-delay and memory usage
  in $O(\log \card{T})$.
\end{theorem}

This theorem allows us to compute the required OR-index.
Indeed, we can compute the OR-components
of~$C$ in linear time, go over each OR-component $K$, and apply the theorem to the
\emph{reverse} of~$K$ (in which exits are leaves), which is still a multitree.
  The result over all OR-components is an index that allows us, 
given any OR-gate $g$ in~$C$, to enumerate the exits of~$g$ with 
constant delay and with the claimed memory usage.
This computation is linear-time overall, and
concludes the preprocessing of our input circuit.
All that remains is to prove Theorem~\ref{thm:multitreeidx}, which we do in the rest of this appendix.

Given a multitree $T$, we will show how to compute in linear
time a multitree $Q(T)$ labeled with leaves of $T$ and a mapping $q$
  from $T$ to $Q(T)$ such that the \emph{leaves} reachable from a node $n\in T$
  correspond (in a one-to-one correspondence) to the labels of the \emph{nodes}
reachable from~$q(n)\in Q(T)$. This ensures that, by enumerating the labels of
  the \emph{nodes} reachable from~$q(n)$ in~$Q(T)$, we enumerate the
  \emph{leaves} reachable from
$n$ in~$T$.

We then show that this second task is easy, as the nodes of a multitree $T$ reachable from a node
$n\in T$ can be enumerated in constant delay with a simple tree traversal.

Finally, we improve on the tree traversal so that the memory usage is
logarithmic in the size of the multitree.

\subparagraph{Transformation of $T$ into $Q(T)$.}
Let $T$ be a multitree, and let us explain how to construct~$Q(T)$.
  We may assume without loss of generality that $T$ is binary. Indeed, if this
  is not the case, we simply
consider all nodes of~$T$ with more than two children and replace them by binary trees in the obvious
way. This transformation does not change the leaves that are reachable from the
original nodes, so it suffices to solve our enumeration problem on the new binary tree.
  
  We create $Q(T)$ by a bottom-up traversal of~$T$, and consider every node~$n$ of~$T$ from the
leaves to the root:
  
  \begin{itemize}
      
    \item If $n$ is a leaf, we introduce a leaf node
      $q(n)$ in~$Q(T)$.
      \item If $n$ is an internal node with a single
        child $n'$, we introduce a node $q(n)$ in~$Q(T)$ and connect it as a
        parent of the children
      of $q(n')$ in~$Q(T)$ if they exist (note that they must already have been
      constructed).
      \item If $n$ is an internal node with two children $n_1$ and $n_2$, we
        introduce two nodes $q(n)$ and $c(n)$ in~$Q(T)$. We connect $q(n)$ as a parent of
      $c(n)$, and connect $c(n)$ as a parent of the children of $q(n_1)$ and
      $q(n_2)$ in~$Q(T)$ if they exist (again, they have already been
      constructed).
  \end{itemize}
  
  This completes the construction of $Q(T)$. Note that multiple internal nodes
  of $Q(T)$ may share the same children, so it is not generally a tree.
  
  We show that $Q(T)$ is a binary
  multitree. Indeed, it is immediate to see that whenever there is an edge
  in~$Q(T)$ from a node $c(n)$ or $q(n)$ to a node $c(n')$ or $q(n')$, then
  either $n = n'$ or $n'$ is a descendant of~$n$ in~$T$. Hence, $Q(T)$ is
  acyclic, and if there is a path from $c(n)$ or $q(n)$ to $c(n')$ or $q(n')$
  in~$Q(T)$, then there is a path from $n$ to~$n'$ in~$T$, so any violation of
  the fact that $Q(T)$ is a multitree would imply a violation in~$T$. This shows
  that $Q(T)$ is a multitree (but note that it is generally not a tree). Further, 
  it is immediate to show by induction that all nodes of the
  form $q(n)$ have at most one child, and then the nodes of the form $c(n)$ have
  at most two children, so indeed $Q(T)$ is binary.

We now describe how to label each node $n'$ of~$Q(T)$ with a leaf of $T$, which
  we write $\lambda(n')$.
  Our construction will ensure the following property: for any node $n\in T$,
  the \emph{leaves} reachable from
$n$ in~$T$ are in a one-to-one correspondence with the labels of \emph{nodes}
reachable from~$q(n)$ in~$T'$. More precisely, for each leaf $\ell$ of $T$
reachable from $n$, there is exactly one node reachable from~$q(n)$ in~$T'$ such
that $\lambda(q(n)) = \ell$,
and, conversely, for every node $n'$ reachable from $q(n)$ in~$T'$, the leaf
$\lambda(n')$ is reachable from~$n$ in~$T$.
We describe the construction in a bottom-up
fashion on nodes $n$ of~$T$, and show that the property is verified for~$n$:

  \begin{itemize}
    \item If $n$ is a leaf, we set $\lambda(q(n)) \colonequals n$. This
      clearly satisfies the property.

    \item If $n$ is an internal node with a single child $n'$,
      we set $\lambda(q(n)) \colonequals \lambda(q(n'))$, which was defined
      before. Since we reach the same leaves from $n$ and $n'$ in $T$, the
      property is satisfied by induction.

    \item If $n$ is an internal node with two children $n_1$ and $n_2$,
      we set $\lambda(q(n)) \colonequals \lambda(q(n_1))$, and 
      $\lambda(c(n)) \colonequals \lambda(q(n_2))$, which were both defined
      before. We now explain why this is correct.
      The set of leaves reachable from $n$ in~$T$ is the union of the leaves
reachable from $n_1$ and $n_2$, and this union is disjoint because $T$ is a
      multitree. Now, the set of nodes reachable from $q(n)$ in~$T'$
contains $q(n)$, $c(n)$, and the nodes reachable from $q(n_1)$ or
$q(n_2)$ except $q(n_1)$ and $q(n_2)$ themselves. So our choice of labels
      clearly guarantees the desired property by induction.
  \end{itemize}

\subparagraph*{Enumeration phase.}
Given a node $n\in T$, thanks to the property of $Q(T)$ that we just showed, 
we can enumerate the leaves reachable from $n$ in~$T$ simply by 
traversing the tree rooted in $q(n)$. Our enumeration state is a
stack $\mathcal{S}$ of nodes in the multitree $Q(T)$ that have yet to be processed. At
the beginning of the enumeration, $\mathcal{S}=\{q(n)\}$. At
each step of the enumeration, we pop a node $n'$ of~$T'$ from
$\mathcal{S}$, push the children of $n'$ (if any) back into
$\mathcal{S}$, and then output $\lambda(n')$.

The stack $\mathcal{S}$ can be implemented with a linked list, so that we can
push and pop elements in constant time. It is immediate that this algorithm
can indeed enumerate in constant delay the labels of the nodes reachable from an
node $q(n)$ of~$Q(T)$. So we have solved our initial enumeration problem on~$T$:
given a node $n \in T$, we enumerate the nodes reachable from~$q(n)$ in~$Q(T)$ as
we explained, and by the property that we showed, the process enumerates exactly
the leaves of~$T$ reachable from~$n$.

\subparagraph*{Memory usage.}
We now explain how to refine the preprocessing and enumeration process to
satisfy the logarithmic memory bound.
We define the \emph{weight} $w(n)$ of a node $n$ in a multitree to be the number of
nodes reachable from a node $n$, including $n$ itself.

The memory usage of our enumeration algorithm is the maximum size of the state
maintained during the enumeration, i.e., the maximum size reached by
$\mathcal{S}$. Given a node $n\in T$, if the tree rooted in $q(n)$ is
unbalanced then $\mathcal{S}$ may contain as many as $w(q(n))/2$
nodes. We now show how to get a tighter, logarithmic bound on memory
usage by choosing the order in which we traverse $Q(T)$.

We first pre-compute the weight of all nodes of~$Q(T)$ in linear time in a
bottom-up fashion, as part of our preprocessing. Now, at each step of the
enumeration, we pop a node from $\mathcal{S}$ (which is the last inserted still
in $\mathcal{S}$) and then push its children onto $\mathcal{S}$. When there are
two children, we make sure that the child with greater weight is pushed first.

We claim that at every step of the enumeration for a node $q(n)$ of~$T'$
corresponding to $n\in T$, the weight of a node in $\mathcal{S}$ is greater
or equal to the sum of the weights of nodes in~$\mathcal{S}$ that
were inserted afterwards, i.e., that precede the node in~$\mathcal{S}$. We show 
the claim by induction along the enumeration process:

\begin{itemize}

\item At the beginning of the enumeration $\mathcal{S}$ contains only
$q(n)$ and the claim is vacuous.

\item At each step of the enumeration, 
we pop the first node and we push back its children, of which there are at most
    two. The
property holds for all the nodes of the new stack that already existed in the
    old stack, since the total
weight of the nodes we push back (i.e. the children) is the weight of
the popped node minus one. The property also holds for the newly added
nodes, because we add the node with bigger weight first.

\end{itemize}

Hence, we have shown our claim by induction. Thus, let us consider any point of the
enumeration algorithm, write the stack $\calS = (s_1, \ldots, s_p)$, 
and show that $p = O(\log\card{T})$. We assume in particular that $p \geq 2$,
otherwise there is nothing to show. Let us define a sequence $(T_i)$ by $T_1
\colonequals 1$ and $T_{i+1}\colonequals\sum_{j\leq i} T_j$: it is clear by induction from
our previous claim that, for all $1 \leq q \leq p$, we have $w(s_q) \geq T_q$.
Now, it is easy to see that $T_{i}=2^{i-2}$ for $i\geq 2$, so we have $w(s_p)
\geq 2^{p-2}$. Remember now that the weight of a node in $\mathcal{S}$ cannot
exceed $w(q(n))$, because all nodes in $\mathcal{S}$ are reachable from $q(n)$.
So we must have $w(s_p) \leq w(q(n))$, and $w(q(n)) \geq 2^{p-2}$. This clearly
implies that $p = O(\log(w(q(n)))$, in particular $p = O(\log\card{T})$. Hence,
the stack is always of size logarithmic in~$\card{T}$, which proves the memory
usage claim.
\end{toappendix}

\section{Enumerating Assignments}
\label{sec:enum}
\begin{toappendix}
  \label{apx:enum}
\end{toappendix}

We have described in the previous sections our linear-time preprocessing on the
input circuit: this produces a normal d-DNNF $C$
together with an OR-index, and we wish to enumerate its assignments $S(C)$ in
zero-suppressed semantics. In this section, we show that we can 
enumerate the elements of $S(C)$, producing each assignment $t$ with delay
$O(\card{t})$.

To prove this, we will go back to our 
definition of zero-suppressed semantics in Section~\ref{sec:reducing},
namely, the minimal valuations of the traces of~$C$ (recall
Definition~\ref{def:trace}).
We will proceed in two steps. First, we use our preprocessing and the OR-index
to show
an efficient enumeration scheme for
the traces of~$C$,
in a compact representation called \emph{compressed traces}.
Second, we show how to enumerate efficiently
the minimal valuations of a compressed trace.

\subparagraph*{Compressed traces.}
\begin{toappendix}
  \subsection{Compressed Traces}
\end{toappendix}
We cannot enumerate traces directly because they can be arbitrarily large (e.g.,
contain long
paths of OR-gates) even for assignments of small weight. We accordingly define \emph{compressed traces} as a variant
of traces that collapse such paths:

\begin{definition}
  \label{def:comptrace}
  An \emph{OR-path} of a monotone augmented circuit $C = (G, W, \mu, g_0)$ is a
  path from~$g \in G$ to~$g' \in G$ where all intermediate gates are OR-gates;
  in particular if $(g, g') \in W$ then there is an OR-path from $g$ to~$g'$.
  A \emph{compressed upward tree} of~$C$
  is a pair $(G', W')$ where $G' \subseteq G$ and where $W' \subseteq G' \times
  G'$ is such that for each $(g, g') \in W'$ there is an OR-path from $g$
  to~$g'$: we require that $T$ is a rooted tree up to reversing the
  direction of the edges.
  $T$ is a
  \emph{compressed partial trace} if its internal gates are AND-gates and
  OR-gates such that:
  \begin{itemize}
    \item for every AND-gate $g$ in~$T$, \emph{all} its inputs in~$C$ are
      children of~$g$ in~$T$;
    \item for every exit $g$ in~$T$ (it is an OR-gate), its one input in~$C$ is a child
      of~$g$ in~$T$;
    \item for every non-exit OR-gate $g$ in~$T$, \emph{exactly one}
      of its \emph{exits} $g'$
      in~$C$ is a child of~$g$ in~$T$.
  \end{itemize}
  We write $\card{T} \colonequals \card{G'}$.
  We call $T$ a
  \emph{compressed trace} of~$C$ if its root is~$g_0$.
  The \emph{minimal
  valuations} of a compressed trace are defined like for
  non-compressed traces (Definition~\ref{def:minval}).
\end{definition}

The use of compressed traces is that their size is linear in that of
their minimal valuations:

\begin{lemmarep}
  \label{lem:compress}
  For any compressed trace $T$ of a normal circuit $C$ and
  minimal valuation $\nu$ for~$T$ and~$C$, we have 
  $\card{T} \leq 6 \cdot \card{\nu}$.
\end{lemmarep}

\begin{proof}
  First observe that, as $C$ is $\semp$-pruned, $C$ (hence $T$) cannot contain
  AND-gates with no children, or range gates labeled $=0$ or $\geq 0$. Hence,
  each leaf of~$T$ is either a variable gate or a range gate capturing a non-empty set. 
  Remember further that, as $C$ has a compatible order, no two
  leaves can share a common variable. Hence, each leaf of~$T$ contributes at
  least one to the Hamming weight of a minimal valuation~$\nu$, so that, letting $n$
  be the number of leaves of~$T$, we have $\card{\nu} \geq n$.

  As $C$ is arity-two and collapsed, each AND-gate of~$T$ has exactly two
  children, and by definition of a compressed trace each OR-gate of~$T$ has
  exactly one child. Letting $n'$ be the number of AND-gates in~$T$, it is then
  clear that $n' = n-1$. Call AND-gates, variable gates, and range
  gates \emph{useful}: their number is $n + n'$.
  It suffices to show that the number $n''$ of OR-gates
  of~$T$ is at most $2 (n' + n)$. This follows if we can show that, for each useful
  gate, one of its parent, grandparent, and great-grandparent in~$T$ is also
  useful (or is undefined, in the case of the root). Indeed, this implies that $n'' \leq
  2 (n' + n)$ because, if each useful gate covers its parent and grandparent,
  this guarantees that all non-useful gates (namely, all OR-gates) are covered.
  The reason why a parent, grandparent, or great-grandparent of a useful gate must 
  be useful is that, whenever $T$ contains an OR-gate, it is
  either an exit and then its one child is not an OR-gate so it is useful, or it
  is not an exit, in which case its one child is an exit. So we have shown the
  desired inequality, which concludes the proof.
\end{proof}

From a trace $T$ in a normal d-DNNF $C$, we can clearly define a
compressed trace $T'$ with the same leaves, as follows. Whenever $T$ contains an
OR-gate $g$ whose parent gate $g'$ in~$T$ is not an OR-gate (or when $g$ is the
root of~$T$), as $g$ cannot be an exit, we know that 
there is a OR-path in~$T$ from $g$ to an exit $g''$ of~$g$ in its OR-component.
We ``compress'' this OR-path in~$T'$ as an edge from $g$ to~$g''$.
Conversely, given a compressed trace~$T'$, we can fill it to a trace~$T$ with
the same leaves, by replacing each edge from $g$ to $g'$ by a witnessing
OR-path; and there is only one way to do so because OR-components in~$C$ are
multitrees (Lemma~\ref{lem:multitrees}). Hence, there is a bijection between
traces and compressed traces that preserves the set of leaves. As the minimal
valuations of traces and compressed traces are defined in the same way from
their set of leaves, we can simply enumerate compressed traces instead
of traces. 

The following shows that we can perform enumeration of the compressed traces
efficiently:

\begin{propositionrep}
  \label{prp:enumtraces}
  Given a normal d-DNNF $C$ with its OR-index, we can
  enumerate its compressed traces, with the delay to produce each compressed
  trace $T$ being in $O(\card{T})$.
\end{propositionrep}

In particular, if all compressed traces have constant size, then the delay is
constant.

\begin{proofsketch}
  At each AND-gate, we
enumerate 
  the lexicographic product of the partial traces of its two children; at each
OR-gate, we enumerate its exits using the OR-index.
\end{proofsketch}

\begin{proof}
  We define inductively an algorithm to enumerate the sequence of partial
  compressed traces in~$C$ rooted at a gate~$g$ as follows:

  \begin{itemize}
    \item If $g$ is a variable, produce the one element of its singleton
      sequence of compressed traces and halt immediately.
    \item If $g$ is an OR-gate, enumerate with constant delay its sequence of
      reachable exits, using the precomputed OR-index: as the circuit is
      normal, this sequence is non-empty. For each reachable exit $g'$,
      letting $g''$ be its one input gate, enumerate the sequence $T(g'')$ of
      partial compressed
      traces rooted at~$g''$. For each such compressed trace $C$, produce $C
      \cup \{g, g'\}$ (the union is disjoint). Halt when the enumeration of reachable exits has halted
      with the last such gate $g'$, and the enumeration of partial compressed
      traces rooted at~$g'$ has halted.
    \item If $g$ is an AND-gate, as the circuit is normal it has exactly
      two inputs. Enumerate the sequence of
      partial compressed traces $T(g_1)$ rooted at its first input~$g_1$. For each trace $C_1
      \in T(g_1)$, enumerate
      the sequence $T(g_2)$ of partial compressed traces rooted at its
      second input~$g_2$.
      For each such trace $C_2$, produce $\{g\} \cup C_1 \cup C_2$ (the unions
      are disjoint). Halt
      when the enumeration of compressed traces of $T(g_1)$ has halted 
      with the last such trace $C_1$ and the enumeration of $T(g_2)$ has also
      halted.
  \end{itemize}

  Running the algorithm on~$C$ is simply running it on the output
  gate~$g_0$.

  We claim that the delay of this algorithm when producing a compressed
  trace~$T$ is in~$O(\card{T})$. To see why, observe that the state when we
  start to enumerate the next valuation consists of gates of~$C$ where the
  enumeration has not yet halted, or (in the case of left children of AND-gates)
  where enumeration has halted but where the previous compressed trace will be
  reused in full. At each node that we consider in the algorithm, we perform a
  constant amount of computation (in particular, for OR-gates, we use the
  OR-index), and then we output that gate as part of the compressed trace.
  Hence, the algorithm performs a constant number of steps at a set of gates
  which is a subset of the gates of the compressed trace which is output, so the
  claim holds. (In particular, when the enumeration at one gate halts, the end
  of the computation at that gate is accounted as part of the last compressed
  trace using a recursive call at that gate, but it is not considered when
  producing the next compressed trace (which may not include that gate).)
\end{proof}

\subparagraph*{Enumerating valuations of a compressed trace.}
\begin{toappendix}
  \subsection{Enumerating Valuations of a Compressed Trace}
\end{toappendix}

We now show how, given a compressed trace~$T$, we can enumerate its minimal
valuations (recall Definition~\ref{def:minval}). Restricting
our attention to the leaves of~$T$, we can rephrase our problem in the following
way:

\begin{definition}
  \label{def:assenumprob}
  The \emph{assignment enumeration problem} for a total order $<$ on gates
  $C_\var$ is as follows: given pairwise disjoint intervals $[g_1^-,
  g_1^+], \ldots, [g_n^-, g_n^+]$, and cardinality constraints ${\bowtie_1 i_i},
  \ldots, {\bowtie_n i_n}$, where $0 < i_j \leq \card{[g_j^-, g_j^+]}$ and ${\bowtie_j} \in
  \{{=}, {\geq}\}$, enumerate the values of the products $t_1 \times \cdots
  \times t_n$ for all the assignments of the $t_j \subseteq [g_j^-, g_j^+]$ such
  that $\card{t_j} \bowtie_j i_j$ for all~$j$.
\end{definition}

Indeed, remember that, as $C$ is $\semp$-pruned, the leaves of~$T$ consist of
variables and range gates, and
their intervals are pairwise disjoint
thanks to decomposability.
A ${\bowtie i}$-gate with inputs $g^-, g^+$ codes the interval $[g^-, g^+]$ with
cardinality constraint ${\bowtie i}$, and a variable~$g$ simply codes 
$[g, g]$ with 
constraint ${=1}$. Further, thanks to
$\semp$-pruning, we know that no range gate is labeled with ${=0}$ or ${\geq
0}$, and thanks to $\emptyset$-pruning, we know that no range gate is labeled
with an infeasible cardinality constraint. We claim:

\begin{propositionrep}
  \label{prp:enumsol}
  We can enumerate the solutions to the assignment enumeration problem for~$<$
  on~$C_\var$, with each solution $t$ being produced with delay
  linear in its size $\card{t}$.
\end{propositionrep}

Again, this is constant-delay when all solutions have size bounded by a
constant.

\begin{proofsketch}
  We enumerate the possible assignments of weights to intervals with
  constant-delay, to reduce to the case where all cardinality constraints are
  equalities. We then enumerate the assignments in lexicographic order, using an
  existing scheme~\cite[Section 7.2.1.3]{Knuth05} to enumerate the assignments
  in each interval.
\end{proofsketch}

We have now concluded the proof of Theorem~\ref{thm:main}
and~\ref{thm:constant}: refer back to Figure~\ref{fig:schema} for an overview of the proof
of Theorem~\ref{thm:main}.
Given our input d-DNNF $C$ and v-tree $T$ rewritten to
a compatible order, we rewrite $C$ to an equivalent normal d-DNNF and compute the OR-index. We
then enumerate compressed traces, and enumerate the valuations for each trace.
The proof of Theorem~\ref{thm:constant} is the same except that we additionally
use Proposition~\ref{prp:homogenize} before Proposition~\ref{prp:normalize} to
restrict to valuations of Hamming weight~$\leq k$.

\begin{toappendix}
To prove this result, it will be convenient to enumerate assignments following
the \emph{lexicographic product} of the individual orders:

\begin{definition}
  \label{def:lexicographic}
Given two sets $S_1, S_2$ with orders $\le_1, \le_2$, the \emph{lexicographic product} $\le_1\times \le_2$ on $S_1\times S_2$ is defined by $(a_1, a_2) (\le_1\times \le_2) (b_1, b_2)$ if and only if 
\begin{itemize}
 \item $(a_1 <_1 b_1)$, or
 \item $a_1 = b_1$ and $(a_2 \le_2 b_2)$.
\end{itemize}
The lexicographic product of two totally ordered sets is clearly a total order,
  and the lexicographic product operation is clearly associative, so
this definition extends to an arbitrary number of sets and yields a total order.
\end{definition}

We show the following general lemma about enumeration in the
lexicographic order:

\begin{lemmarep}
  \label{lem:lexicographic}
  Let $S_1, \ldots, S_n$ be non-empty sets that do not contain the empty assignment, such that
  the elements of each $S_i$ can be enumerated in some total order $\le_i$, each
  element being produced with delay linear in its size. Then the elements of the
  product $S_1 \times \cdots \times S_n$ can be enumerated in the lexicographic
  order $\le_1 \times \cdots \times \le_n$, each element being produced with
  delay linear in its total size.
\end{lemmarep}

\begin{proof}
  We run the enumeration algorithm for each $S_j$. To produce the first
  enumeration result, we enumerate the
  first element of each~$S_j$ with its algorithm, and we find the largest $1
  \leq j \leq n$ such that 
  the element of $S_j$ that we enumerated is not the last one: this
  obeys the delay bound, because, as the $S_j$ do not contain the empty assignment,
  the time required to iterate over all the $S_j$ is linear in the total size of
  the enumerated solution.
  
  Now, at each stage of the global enumeration algorithm, we remember the last
  element of the product that we enumerated, the corresponding enumeration
  state in each~$S_j$, and the largest~$1 \leq j \leq n$ such that the element
  of~$S_j$ that we enumerated is not the last one. To produce the next element,
  enumerate the next element $e$ of this~$S_j$, and then compute first element $e_{j^{\rightarrow}}$  
  for $S_{j^\rightarrow}$ for $j < j^\rightarrow \leq n$
  as when producing the first enumeration result (this may be empty if $j = n$).
  Finally, we go over all $S_j$ to update our value for~$j$. We then
  produce our enumeration result: it is composed of the element of the product
  of the $S_{j^\leftarrow}$ for $1 \leq j^\leftarrow < j$ that we had 
  enumerated in the round before, of the element $e \in S_j$ that we just enumerated, 
  and the $e_{j^\rightarrow}$ for $j < j^\rightarrow \leq n$. The delay of this is the
  delay of writing the solution, which is linear in its total size, plus the
  delay of going over the $S_j$, which is linear in the total size as above, and
  the delay of enumerating $e \in S_j$ and the $e_{j^{\rightarrow}}$, which is less than the total size again,
  so we obey the bound.
\end{proof}

We will use Lemma~\ref{lem:lexicographic} to enumerate the solutions to the
  assignment enumeration problem (recall Definition~\ref{def:assenumprob}),
and we will do so in two steps. We will first reduce to the case where all
constraints~$\bowtie_j$ are equalities. Then we will enumerate valuations in
this case.

To reduce to equalities, we will define the \emph{range} $R_j$ of the interval
$[g_j^-, g_j^+]$ for $1 \leq j \leq n$ as the singleton $\{i_j\}$ if $\bowtie_j$
is~${=}$, and the range $\{i_j, i_j+1, \ldots, \card{[g_j^-, g_j^+]}\}$ if
$\bowtie_j$ is~${\geq}$;
note that this set is non-empty. The \emph{range} $R$ of the product is simply
$R_1 \times \cdots \times R_n$. We can talk of an element $t_1 \times \cdots
\times t_n$ of the product of the intervals as \emph{realizing} the vector
$(\card{t_1}, \ldots, \card{t_n})$ of~$R$, which we call a \emph{histogram}:
clearly all such elements realize a histogram of~$R$ (so the values of~$R$
partition the assignments), and conversely every value of~$R$ is the histogram
of some element (i.e., the classes of the partition are non-empty).
Thus, to prove Proposition~\ref{prp:enumsol}, we can enumerate the
histograms of the range~$R$, and then enumerate the assignments corresponding to
this histogram.

It is clear that, for each range $R_j$, we can enumerate the integers that it
contains with constant delay since we are working in a RAM model. Hence, by
Lemma~\ref{lem:lexicographic}, we can enumerate the histograms of~$R$ with
delay linear in $n$, i.e., the number of entries in the histograms. Note now that $n$
is always less than the size of any assignment that realizes it because the $i_j$ and 
thus the number of inputs chosen from each interval $[g_j^-, g_j^+]$ is strictly positive, so
the delay when enumerating a histogram is within the allowed delay to
enumerate an assignment. Note that, as each histogram is realized by at least
one assignment, the delay when enumerating a histogram is paid at most once
when enumerating an assignment.

Hence, it suffices to study the enumeration of assignments that satisfy a fixed
histogram, i.e., where $\bowtie_j$ is $=$ for each $1 \leq j \leq n$.
Let $\binom{S}{q}$ for a set $S$ and a non-negative integer $q$ denote the set of all subsets of size $q$ of $S$. We make the following observation.

\begin{observationrep}\label{obs:lexicographicproduct}
 Let $[g_1^-, g_1^+], \ldots, [g_n^-, g_n^+]$ and ${= i_i}, \ldots, {= i_n}$ be an instance of the assignment enumeration problem where all cardinality constraints are equalities. Then the assignments to be enumerated on the given instance are exactly $\binom{[g_1^-, g_1^+]}{i_1}\times \ldots \times \binom{[g_n^-, g_n^+]}{i_n}$.
\end{observationrep}
\begin{proof}
  By definition of the assignment enumeration problem, when choosing for each $j\in [n]$ an assignment $t_j$ of size $i_j$ from $[g_1^-, g_1^+]$, we have that $t_1\times \ldots \times t_n$ has to be enumerated. Conversely, every assignment $a$ that has to be enumerated decomposes as $a_1\times \ldots \times a_n$ with $\card{a_j} = i_j$ and thus $t\in \binom{[g_1^-, g_1^+]}{i_1}\times \ldots \times \binom{[g_n^-, g_n^+]}{i_n}$.
\end{proof}

Hence, applying Lemma~\ref{lem:lexicographic} again, it suffices to argue that
we can enumerate the elements of the $\binom{[g_j^-,g_j^+}{i_j}$ in delay linear
in the size of the produced elements, i.e., linear in~$i_j$. We will see the
elements of $[g_j^-,g_j^+]$ as ordered according to~$<$, which allows us to
define a lexicographic order on $\binom{[g_j^-,g_j^+}{i_j}$. It is then known
that we can enumerate such elements with delay linear in~$i_j$; 
we refer the reader to e.g.~\cite[Section 7.2.1.3]{Knuth05} where implicitly the following is shown.

\begin{proposition}\label{prop:enumeratecombinations}
 Given a set $S$ of $p$ ordered elements and $q \in \mathbb{N}$, the following tasks can be performed in time $O(q)$:
 \begin{itemize}
  \item compute the lexicographically minimal combination of $q$ elements from $S$, and
  \item given a combination of $q$ elements from $S$, compute the lexicographically next such combination if it exists.
 \end{itemize}
\end{proposition}

Hence, by Lemma~\ref{lem:lexicographic} we can enumerate the assignments
satisfying a histogram, producing each assignment with delay linear in its total
size. This concludes the proof of Proposition~\ref{prp:enumsol}.
\end{toappendix}

\begin{toappendix}
  \subsection{Putting Things Together}
  \label{apx:putting}
We are now ready to  put things together to prove our main results. We first
  show:
  
\begin{proposition}
  \label{prp:enum}
  Given a normal d-DNNF $C$ with its OR-index, we can
  enumerate the elements of $S(C)$, producing each assignment $t$ with delay
  $O(\card{t})$.
\end{proposition}

\begin{proof}
The enumeration algorithm consists of two nested loops: In the outer loop, we enumerate the compressed traces $T$ of $C$ with the help of Proposition~\ref{prp:enumtraces}. In the inner loop, we enumerate for each $T$ the satisfying assignments with Proposition~\ref{prp:enumsol}. Since~$C$ is deterministic, each satisfying assignments of $C$ is captured by exactly once compressed trace. Consequently, we enumerate every satisfying assignments of $C$ exactly once, so the algorithm is correct.

To analyze the delay of the algorithm, note that, to enumerate a valuation~$\nu$, in the worst case we have to first enumerate the next compressed trace $T$ of $C$ and then compute the valuation~$\nu$ as a valuation of $T$. The first part takes time~$O(\card{T})$ by Proposition~\ref{prp:enumtraces} which by Lemma~\ref{lem:compress} is $O(\card{\nu})$. The second part takes time $O(|\nu|)$ by Proposition~\ref{prp:enumsol}. So the overall delay to produce $\nu$ is $O(\card{\nu})$ as claimed.
\end{proof}

We are now ready to prove our first main result:

\begin{proof}[Proof of Theorem~\ref{thm:main}]
 Given $C$, we first deal with two special cases. We first check if $C$ has any satisfying assignments. If not, we are done at this point and stop. Note that this consistency check can be done in linear time~\cite{Darwiche01}.
 
 The second special case is that we check if $C$ is satisfied exactly by $\ssemp$. If so, we print out $\semp$ and are done. This test can also be done in linear time as follows: First check if $\semp$ satisfies $C$. This can be done in linear time by substituting all inputs by $0$ and then evaluating $C$. Afterwards, we check if $C$ is satisfied by exactly one valuation. Since satisfying assignments of a d-DNNF can be counted in linear time~\cite{DarwicheM02}, this is also a linear time test.
 
 In the remainder of the proof, we may now assume that the set $S$ of valuations satisfying~$C$ is such that $S\ne \emptyset$ and $S\ne \ssemp$.
 
 We now infer a compatible order $<$ for $C$. As discussed in Section~\ref{sec:reducing}, this is easy to do in linear time, assuming we are given a v-tree.
  Next, we proceed with Proposition~\ref{prp:reducing} to compute a monotone $0$-augmented d-DNNF $C'$ in zero-suppressed semantics having $<$ as a compatible order such that $S(C')=S$.
  We then use Proposition~\ref{prp:normalize} to compute a $1$-normal d-DNNF $C''$ which has~$<$ as a compatible order and is such that $S(C'')= S(C')\setminus \ssemp$.
  Finally, we compute the OR-index of $C''$ with Theorem~\ref{thm:multitreeidx}.
 
 Before we start the enumeration phase, we check if $\semp$ satisfies $C$. If so, we enumerate $\semp$ as the first valuation. Afterwards, we use Proposition~\ref{prp:enum} to enumerate the valuations in $S\setminus \ssemp$.
 
 By inspection of the individual results used in this algorithm, it is obvious that the satisfying assignments of $C$ are correctly enumerated. Moreover, the linear runtime bound on the preprocessing follows by the fact that all individual steps can be performed in time linear in their input size. The bound on the enumeration delay follows directly from Proposition~\ref{prp:enum}.
\end{proof}

The proof of Theorem~\ref{thm:constant} is identical to that of Theorem~\ref{thm:main} except for the fact that we make an additional preprocessing step. After using Proposition~\ref{prp:reducing}, we compute a circuit $C^k$ that is satisfied exactly by the satisfying assignments of $C$ with Hamming weight at most~$k$ with Proposition~\ref{prp:homogenize}. We then proceed as in the proof of Theorem~\ref{thm:main}. Note that this slightly increases the runtime of the preprocessing from $O(|C|)$ to $O(k^2\cdot |C|)$.
\end{toappendix}

\section{Applications}
\label{sec:applications}
We now present two applications of our main results.
Our first application recaptures the well-known enumeration results for MSO
queries on trees~\cite{bagan2006mso,kazana2013query}. The second application
describes links to factorized databases
and strengthens the enumeration result of~\cite{olteanu2015size}.

\subparagraph*{MSO enumeration.}
\begin{toappendix}
  \subsection{Computing Circuit Representations of MSO Answers}
  \label{apx:msofac}
\end{toappendix}
Recall that the class of \emph{monadic second-order} formulae (MSO)
consists of first-order logical formulae extended with quantification over sets, see e.g.~\cite{Libkin04}.
The \emph{enumeration problem} for a fixed MSO formula~$\phi(X_1, \ldots, X_k)$
with free second-order variables,
given a structure $I$, is to enumerate the \emph{answers} of~$\phi$ on~$I$,
i.e., the $k$-tuples $(B_1, \ldots, B_k)$ 
of subsets of the domain of~$I$ such that $I$ satisfies $\phi(B_1, \ldots,
B_k)$. We measure the \emph{data complexity} of this task, i.e., its complexity
in the input structure, with the query being fixed.

It was shown by Bagan~\cite{bagan2006mso} that MSO query enumeration on
\emph{trees} and \emph{bounded treewidth structures} can be performed with
linear-time preprocessing and delay linear in each MSO assignment; in particular, if
the free variables of the formula are first-order, then the delay is constant.
This latter result was later re-proven by Kazana and Segoufin
\cite{kazana2013enumeration}. We show how to recapture this theorem from our
main results. From the results of Courcelle and standard techniques (see, e.g.,
\cite{kazana2013query}, Theorem~6.3.1 and Section~6.3.2), we restrict to binary
trees.

\begin{definition}
  Let $\Gamma$ be a finite alphabet.
  A \emph{$\Gamma$-tree} $T$ is a rooted unordered binary tree
  where each node $n \in T$ carries a label in $\Gamma$. We abuse notation and identify $T$
  to its node set. \emph{MSO formulae on $\Gamma$-trees} are written on
  the signature consisting of one binary predicate for the edge relation and 
  unary predicates for each label of~$\Gamma$.
\end{definition}

Let $\phi(X_1, \ldots, X_k)$ be an MSO formula on $\Gamma$-trees, and let $T$ be
a $\Gamma$-tree. We will show our enumeration result by building a structured
circuit capturing
the \emph{assignments}
of~$\phi$ on~$T$:

\begin{definition}
  \label{def:msoassign}
  A \emph{singleton} on~$X_1, \ldots, X_k$ and~$T$ is an expression of the form $\langle X_i : n
\rangle$ with $n \in T$.
An \emph{assignment} on~$X_1, \ldots, X_k$ and~$T$ is a set $S$ of singletons: 
it defines a $k$-tuple $(B^S_1, \ldots, B^S_k)$ of subsets of $T$ by 
setting $B^S_i \colonequals \{n \in T \mid \langle X_i: n \rangle \in S\}$ for
  each~$i$. The \emph{assignments} of~$\phi$ on~$T$ are the assignments $S$ such
  that $T$ satisfies $\phi(B^S_1, \ldots, B^S_k)$.
\end{definition}

We will enumerate assignments instead of answers: this makes no difference
because we can always rewrite each assignment in linear time to the
corresponding answer. We
now state the key result: we can efficiently build circuits (with singletons
as variable gates) that capture the assignments to MSO queries. (While these
circuits are not augmented circuits, they are decomposable, so the definition of
zero-suppressed semantics clearly extends.)

\begin{theoremrep}
  \label{thm:facrep}
  For any fixed MSO formula $\phi(X_1, \ldots, X_k)$ on $\Gamma$-trees,
  given a $\Gamma$-tree $T$, we can build in 
  time $O(\card{T})$ a monotone d-DNNF circuit~$C$ in zero-suppressed semantics
  whose set $S(C)$ of assignments (as in Definition~\ref{def:zss})
  is exactly the set of assignments of~$\phi$ on~$T$.
\end{theoremrep}

\begin{proofsketch}
  We simplify $\phi$ to have a single free variable and limit to assignments on
  leaves as in~\cite{bagan2006mso}, and
  rewrite $\phi$ to a deterministic tree automaton $A$
  using the result of Thatcher and Wright~\cite{thatcher1968generalized}, in
  time independent of~$T$ (though the runtime is generally nonelementary
  in~$\phi$). We then compute our circuit as a variant of the \emph{provenance
  circuits} in our earlier work~\cite{amarilli2015provenance}, observing that it is a d-DNNF
  thanks to the determinism of the automaton as in~\cite{amarilli2016tractable}.
  This second step is in~$O(\card{A} \cdot \card{T})$, so linear in~$T$.
  Appendix~\ref{apx:msofac} gives a self-contained proof.
\end{proofsketch}

Note that the resulting circuit is already in zero-suppressed semantics, and has
no range gates. By continuing as in the proof of Theorem~\ref{thm:main} (for free second-order
variables) or of Theorem~\ref{thm:constant} (for free first-order
variables), we deduce the MSO enumeration results
of~\cite{bagan2006mso,kazana2013enumeration}. Note that, once we have computed
the tree automaton for the query and the circuit representation, our proof of the enumeration result is
completely query-agnostic: we simply apply our enumeration construction on the circuit.
Our proof also does not depend 
on the factorization forest
decomposition theorem of~\cite{colcombet2007combinatorial} used 
by~\cite{kazana2013enumeration}; it
consists only of the simple circuit manipulation and indexing that we presented in
Sections~\ref{sec:normal}--\ref{sec:enum}. Note that the delay is in $O(k \cdot
\card{T})$, with no large hidden constants, and $O(k)$ for first-order variables.

A limitation of our approach is that our memory usage bound includes a
logarithmic factor in~$T$, whereas \cite{bagan2006mso,kazana2013enumeration}
show constant-memory enumeration. However, we can show that the circuit computed
in Theorem~\ref{thm:facrep} satisfies an \emph{upwards-determinism} condition that allow us to
replace the indexing scheme of Theorem~\ref{thm:orindex} (our memory
bottleneck) by a more efficient index. We can thus
reprove the constant-memory enumeration of~\cite{bagan2006mso,kazana2013enumeration} (see
Appendix~\ref{apx:memory}).

\begin{toappendix}
  This appendix section proves Theorem~\ref{thm:facrep}; we later explain in
Appendix~\ref{apx:msoenum} how we
can use this result to deduce MSO enumeration results using our main results. The key
ingredient of the proof of Theorem~\ref{thm:facrep} is our existing construction
for provenance of MSO queries on treelike
instances~\cite{amarilli2015provenance}, using automaton determinism to obtain a
d-DNNF \cite{amarilli2016tractable}. However, for readability, we give a
self-contained proof of this result, which focuses on the case of trees. The
rewritten proof presented here is also useful to show upwards-determinism and
deduce constant memory bounds for enumeration (see Appendix~\ref{apx:memory}).

We introduce some additional notation.
Given a $\Gamma$-tree $T$, 
we will write $\lambda(n)$ to denote the label in~$\Gamma$ of a node $n$ of~$T$;
in other words, the labeling function $\lambda$ is part of the $\Gamma$-tree,
but we do not write it explicitly for brevity.
We will write $\Leaf(T)$ for the set of leaves of a $\Gamma$-tree~$T$. Remember
that we often identify $T$ with its set of nodes when no confusion can ensue.

Further, we will write $\Assign(\phi, T)$ to denote the set of
assignments of an MSO formula $\phi(X_1, \ldots, X_k)$ on $\Gamma$-trees
with free second-order variables on a $\Gamma$-tree $T$, i.e., the set of
assignments $A$ on schema~$\mathbf{X} = X_1, \ldots, X_k$ and domain~$T$ such that $T$ satisfies $\phi(A_1,
\ldots, A_k)$ with the $A_i$ defined as in Definition~\ref{def:msoassign}.

To prove Theorem~\ref{thm:facrep}, somewhat similarly to Sections~3.3.2 and
Sections~3.3.3 of~\cite{bagan2009algorithmes}, it will be useful to assume
that assignments are only considered on leaves, and that the MSO formula only
has one free second-order variable. We will explain how to do this, up to extending
the size of the alphabet.

\begin{definition}
  Let $\Gamma$ be a finite alphabet of labels, let $\mathbf{X} = X_1, \ldots,
  X_k$ be a tuple of second-order variables which we see as labels disjoint
  from~$\Gamma$, and let $\bot$ be a fresh node label. Let $\Gamma^{\mathbf{X}}
  \colonequals \Gamma \cup \{\bot, X_1, \ldots, X_k\}$.

  A \emph{$\Gamma^{\mathbf{X}}$-assignment tree} $(T, \mu)$ is a
  $\Gamma^{\mathbf{X}}$-tree $T$
  and a mapping $\mu$ from
  $\Leaf(T)$ to a domain $\calD$ called the \emph{domain} of the
  assignment tree. We impose the following requirements:
  \begin{itemize}
    \item The labels $X_1, \ldots, X_k$ are used only on leaf nodes, and
  conversely every leaf node carries a label of this set. Formally, we
      require $\Leaf(T) = \{n \in T \mid \lambda(n) \in \{X_1, \ldots, X_k\}\}$.
    \item The mapping $\mu$ is
  computable in constant time, i.e., we can read the image by $\mu$ of a leaf
      node of~$T$ directly from that node.
    \item If $\mu(n) \neq \mu(n')$ for two leaves $n \neq n'$ of~$T$,
  we require that $\lambda(n) \neq \lambda(n')$.
  \end{itemize}
  For any $\Gamma^{\mathbf{X}}$-assignment tree~$T$ and subset $U \subseteq
  \Leaf(T)$,
  the \emph{$\mathbf{X}$-assignment} $\alpha(U)$ of~$U$ is defined as $\{\langle \lambda(n):
  \mu(n) \rangle \mid n \in U\}$. Note that this set is without duplicates
  thanks to our requirement on~$\mu$ above, and it is an assignment on
  schema~$\mathbf{X}$ and domain~$T$.
\end{definition}

We now claim that, up to increase the size of the formula, we can rewrite an MSO
formula so that it has only one free variable and only answers that include
leaves need to be considered:

\begin{lemma}
  \label{lem:simplify}
  For any MSO formula $\phi(X_1, \ldots, X_k)$ on $\Gamma$-trees,
  we can compute an MSO formula $\psi(Y)$ on $\Gamma^{\mathbf{X}}$-trees 
  with one free second-order variable that has the following property:
  given any $\Gamma$-tree $T$, we can compute in linear time a
  \mbox{$\Gamma^{\mathbf{X}}$-assignment} tree
  $(T', \mu)$, whose domain is the nodes of~$T$, such that the
  assignments of~$\phi$ on~$T$ are exactly the
  $\Gamma^{\mathbf{X}}$-assignments of the answers of~$\psi$
  on~$T'$; formally:
  $\Assign(\phi, T) = \{\alpha(U) \mid U \subseteq \Leaf(T'), T' \models
  \psi(U)\}$.
\end{lemma}

\begin{proof}
  We rewrite $\phi(X_1, \ldots, X_k)$ to an MSO formula $\psi(Y)$ on
  $\Gamma^{\mathbf{X}}$-trees,
  by creating the free second-order variable $Y$ and
  replacing each atom of the form~$X_i(x)$ for a first-order variable~$x$ and
  free second-order variable~$X_i$ by 
  $\exists y ~ Y(y) \wedge X_i(y) \wedge \Phi(x, y)$, where $\Phi$ is a
  constant-sized MSO subformula asserting that $y$ is a descendant of~$x$ and the
  path from $x$ to~$y$ in the tree passes only through nodes labeled~$\bot$.

  We now describe the linear-time rewriting of input trees. We rewrite an input $\Gamma$-tree $T$
  to a $\Gamma^{\mathbf{X}}$-assignment tree~$(T', \mu)$
  consisting of a $\Gamma^{\mathbf{X}}$-tree $T'$ and function~$\mu$
  from~$\Leaf(T')$ to~$T$ (written directly on the leaves to ensure
  constant-time computability).
  We do so by adding, for every node $n$ of~$T$, $k$ fresh descendants $n_1,
  \ldots, n_k$ that we connect to~$n$ by a binary tree of fresh nodes labeled $\bot$.
  Each $n_i$ is labeled with~$X_i$ and mapped by~$\mu$
  to~$n$. It is clear that this process runs in linear time, remembering that
  $k$ is a constant. Further, it is clear that $(T', \mu)$ uses $X_i$ only on
  leaf nodes, and exactly on such nodes; and that $\mu$ satisfies the
  requirement that it does not map to the same element of~$T$ two leaves of~$T'$
  carrying the same label.
  
  Last, it is immediate that the answers of~$\psi$ on~$T'$ map to the
  assignments of~$\phi$ on~$T$ in the prescribed way. Indeed, 
  the rewriting of~$\phi$ to~$\psi$ clearly ensures that $U \subseteq \Leaf(T')$ is an
  answer to~$\psi$ iff $(U_1, \ldots, U_n)$ is an answer to~$\phi$, where $U_i$
  contains the nodes of~$T$ whose fresh descendant labeled $X_i$ and connected
  by a $\bot$-path in~$T'$ is in~$U$. This is the case iff
  the $\mathbf{X}$-assignment of~$U$ on~$\mathbf{X}$ and~$T$ is an assignment
  to~$\phi$.
\end{proof}

Thanks to this result, we can restrict our study to MSO formulae $\psi(Y)$ with only one free
variable, and to answers of~$\psi$ that only contain leaves of the tree. 
We will now state a simple lemma that asserts that the interpretation of a free
second-order variable in an MSO formula can always be read off directly from the
labels of the tree. We first introduce some definitions:

\begin{definition}
  A \emph{leaf valuation} $\nu$ of a $\Gamma$-tree $T$
  is a function mapping the nodes of~$\Leaf(T)$ to~$\{0, 1\}$; we will abuse
  notation and see them as valuations of~$T$ by extending them to map every
  internal node to~$0$.
  We write $\LVal(T)$ for the set of leaf valuations of~$T$.

  We write $\overline{\Gamma}$ to mean $\Gamma \times \{0, 1\}$.
  For $\nu \in \LVal(T)$,
  we denote by $\nu(T)$ the
  $\overline{\Gamma}$-tree obtained from~$T$ by relabeling each node $n$ 
  from $\lambda(n)$ to $(\lambda(n), \nu(n))$.
\end{definition}

\begin{lemma}
  \label{lem:boolform}
  Given an MSO formula $\psi(Y)$ on $\Gamma$-trees
  with one free variable, we can compute an MSO formula $\chi$ on $\overline{\Gamma}$-trees 
  with no free variables (i.e., a Boolean formula) that has the following property: for any
  $\Gamma$-tree $T$, 
  for any leaf valuation $\nu \in \LVal(T)$, 
  the $\overline{\Gamma}$-tree $\nu(T)$ satisfies $\chi$ iff 
  $\{\langle Y : n \rangle \mid \nu(n) = 1\}$
  is an assignment of~$\psi(Y)$.
\end{lemma}

\begin{proof}
  We simply rewrite each atom $L(x)$ for a node predicate $L$ of~$\Gamma$ by
  $((L, 0))(x) \vee ((L, 1))(x)$, and we replace atoms $Y(x)$ that use the free
  second-order variable $Y$ with $\bigvee_L ((L, 1))(x)$ 
  for all node predicates $L$ in~$\Gamma$. It is then clear that the additional
  label of a $\overline{\Gamma}$-tree indicates how the free second variable
  should be interpreted.
\end{proof}

Remembering that we are only considering answers to the input MSO
formula~$\psi(Y)$ that consist of leaf nodes, this lemma allows us to assume a
Boolean formula $\chi$ on $\overline{\Gamma}$-trees and to study the leaf
valuations $\nu$ of an input $\Gamma$-tree~$T$ such that the $\chi$ accepts
$\nu(T)$. Our goal is to obtain a circuit which captures these
leaf valuations (represented as assignments) under zero-suppressed semantics. In
other words, the circuit will have variable gates that correspond to the nodes
of~$T$, and its captured set should be exactly the assignments corresponding to
leaf valuations of~$T$ that make it satisfy~$\chi$.

To compute this circuit, we will be going through tree automata. To
this end, it will be simpler to think of automata that read \emph{ordered}
trees, i.e., there is an order on the children of each internal node; we will
define automata accordingly but will ensure that this order is inessential. It
will also be simpler to assume that input trees are \emph{full}, i.e.,
every node has either $0$ or $2$ children. To do this, we can always add a fresh
symbol $\bot'$ to the alphabet, with its two labeled versions $(\bot', 0)$ and
$(\bot', 1)$, and add fresh leaves to $\Gamma$-trees labeled $\bot'$ to make
them full. One would then rewrite the MSO formula to relativize quantification
to nodes that are not labeled $\bot'$ (i.e., do not quantify over them), and add
a constant-sized formula asserting that these nodes are all labeled~$0$ so that
they never occur in assignments.

We thus define deterministic bottom-up tree automata in the standard way:

\begin{definition}
  A \emph{bottom-up deterministic tree automaton} on $\Gamma$-trees that are
  full and ordered (and binary),
  called a $\Gamma$-bDTA for brevity, is a tuple $A = (Q, F, \iota, \delta)$ where:
\begin{enumerate}
    \item $Q$ is a finite set of \emph{states};
    \item $F$ is a subset of~$Q$ called the \emph{accepting states};
    \item $\iota : \Gamma \to Q$
      is an \emph{initialization function}
      which determines the state of the automaton on a leaf node from the label
      of that node;
    \item $\Delta : \Gamma \times Q^2
      \to Q$
      is a \emph{transition function} which determines the state of the
      automaton on an internal node
      from its label and the state of the automaton on its two children.
\end{enumerate}
  As our trees are unordered, we require that the
  order in which the automaton reads the children of a node never matters, i.e., for
  every $l \in \Gamma$ and $q_1, q_2 \in Q$, we have $\delta(l, q_1, q_2) =
  \delta(l, q_2, q_1)$.

  Given a $\Gamma$-tree $T$, we define the \emph{run} of $A$ on~$T$
  as the function $f : T \to Q$ defined by:
\begin{enumerate}
  \item For each leaf $l$ of~$T$, set $f(l) \colonequals
    \iota(\lambda(l))$;
  \item For each internal node $n$ of~$T$ with children $n_1$ and $n_2$, 
    set $f(l) \colonequals \delta(\lambda(n), f(n_1), f(n_2))$.
\end{enumerate}
We say that the bDTA $A$ \emph{accepts} a $\Gamma$-tree $T$ if, letting $n_\r$
  be the root of~$T$, the run $f$ of~$A$ on~$T$ is such that $f(n_\r) \in F$.
\end{definition}

We now use the well-known fact that Boolean MSO formulae on $\Gamma$-trees can
be rewritten to equivalent $\Gamma$-bDTAs, using the standard translation of
Thatcher and
Wright~\cite{thatcher1968generalized} and standard techniques to determinize the
automaton~\cite{tata}:

\begin{theorem}[\cite{thatcher1968generalized}]
  \label{thm:thatcherwright}
  For any tree alphabet $\Gamma$ and Boolean MSO formula $\chi$ on
  $\Gamma$-trees, we can compute a $\Gamma$-bDTA $A$ such that, for any
  $\Gamma$-tree $T$, we have that $T$ satisfies $\chi$ iff $T$ is accepted by~$A$.
\end{theorem}

Having fixed our Boolean formula $\chi$ on $\overline{\Gamma}$-trees, let us
compute accordingly such a \mbox{$\overline{\Gamma}$-bDTA}~$A$. Remember that,
given a $\Gamma$-tree $T$,
we want to compute
a circuit whose captured set under zero-suppressed semantics 
is the set of assignments representing leaf
valuations $\nu$ of~$T$ such that $A$ accepts $\nu(T)$. We call this the
\emph{assignment set} of the automaton $A$ on the tree~$T$. 
The following definition is inspired by the provenance notions
in~\cite{amarilli2015provenance}, but changed to work only on leaves.

\begin{definition}
Let $A$ be a $\overline{\Gamma}$-bDTA, and $T$ be a $\Gamma$-tree.
The 
  \emph{assignment set}
  $\alpha(A, T)$ of~$A$ on~$T$ is the set $\{\alpha(\nu) \mid \nu
\in \LVal(T), A \text{~accepts~} \nu(T)\}$. 
\end{definition}

We then give a construction inspired to Proposition~3.1
of~\cite{amarilli2015provenance}, but rephrased in the terminology of factorized
representations, and simplified by limiting the uncertain labels to leaves. We also
observe that the result is deterministic thanks to the determinism of the
automaton, as in Theorem~6.11 of~\cite{amarilli2016tractable}.

\begin{proposition}
\label{prp:provenance}
For any tree alphabet $\Gamma$, given a $\overline{\Gamma}$-bDTA $A = (Q, F, \iota, \delta)$ and a
  full (binary) $\Gamma$-tree $T$,
  we can compute in time $O(\card{A} \cdot \card{T})$ a
  monotone circuit $C$ which is a d-DNNF in zero-suppressed semantics, such that $S(C) = \alpha(A, T)$.
\end{proposition}

Note that $C$ is not an augmented circuit, but as it is decomposable, 
the set $S(C)$ of assignments of~$C$ in zero-suppressed semantics (in the sense of
Definition~\ref{def:zss}, or Lemma~\ref{lem:botup}) is well-defined (recall Remark~\ref{rmk:zerosuppressed}).

\begin{proof}[Proof of Proposition~\ref{prp:provenance}]
  We compute the circuit
  $C$ in a bottom-up fashion on~$T$.
  We consider each node $n$ of~$T$ with label $\lambda(n) \in \Gamma$.

  If $n$ is a leaf node,
  for $b \in \{0, 1\}$ we let $q_b \colonequals \iota((\lambda(n), b))$, and we
  create the following gates in~$C$:
  
  \begin{itemize}
    \item One OR-gate $g^q_n$ for each $q \in Q$ with the following inputs:
    \begin{itemize}
    \item If $q = q_0$, one AND-gate with no inputs.
    \item If $q = q_1$, one variable gate corresponding to the node $n$
    \end{itemize}
  \end{itemize}

  If $n$ is an internal node with children $n_1$ and $n_2$, we create the
  following gates in~$C$:

  \begin{itemize}
    \item One AND-gate $g^{q_1,q_2}_n$ for each $q_1, q_2 \in Q$ whose
      inputs are $g^{q_1}_{n_1}$ and $g^{q_2}_{n_2}$;
    \item One OR-gate $g^q_n$ for each $q \in Q$ with inputs the
      $g^{q_1,q_2}_n$ for each $q_1, q_2 \in Q$ such that $\delta((\lambda(n),
      0), q_1, q_2) = q$.
  \end{itemize}

  The output gate $g_0$ is a $\vee$-gate of the $g^q_{n_\r}$ for
  $q \in F$, where $n_\r$ is the root of~$T$.

  It is clear that the construction of~$C$ runs in the prescribed time bound, because
  the processing that we perform at each node of~$T$ is linear in $\card{A}$,
  specifically, in the table of the transition function $\delta$ of~$A$.

  It is clear that $C$ is decomposable,
  because AND-gates
  that have inputs are of the form $g^{q_1,q_2}_n$ for internal nodes $n$ of~$T$, in
  which case the inputs are $g^{q_1}_{n_1}$ and~$g^{q_2}_{n_2}$. Now, it is
  immediate that, for $i \in \{1, 2\}$,
  only descendant leaves of~$n_i$ can
  appear in $\bigcup S(g^{q_i}_{n_i})$. As these sets of
  descendant leaves for the two sibling nodes~$n_1$ and $n_2$ are disjoint,
  the decomposability condition is indeed satisfied.

  It is now easy to show the following inductive correctness claim on~$C$: 
  for each $q \in Q$ and $n \in T$
  the assignment set $S(g^q_n)$ captured by the gate $g^q_n$ precisely
  describes
  the leaf valuations $\nu$ of the subtree $T_n$ of~$T$ rooted at~$n$ such that the run
  of~$A$ on $\nu(T_n)$ reaches $q$ on the root node~$n$ of~$\nu(T_n)$. Indeed, for a leaf node $n$
  of~$T$ and for $q \in Q$, the assignments corresponding to the possible
  leaf valuations are $\semp$ and $\{n\}$, and we have $\semp \in
  S(g^q_n)$ iff $q = \iota((\lambda(n), 0)$ and $\{n\} \in S(g^q_n)$ iff $q =
  \iota((\lambda(n), 1))$. For an internal node $n$ of~$T$ with children $n_1$
  and~$n_2$ and $q \in Q$,
  an assignment $a$ corresponding to a leaf valuation $\nu$ belongs to $S(g^q_n)$ iff there is a pair $q_1, q_2 \in
  Q$ of states such that $\delta((\lambda(n), 0), q_1, q_2) = q$ and, for each $i \in \{1, 2\}$,
  the assignment $a_i$ of the restriction $\nu_i$ of~$\nu$ to the subtree $T_{n_i}$ rooted at $n_i$ belongs to
  $S(g^{q_i}_{n_i})$. By induction hypothesis, for any $q_1, q_2 \in Q$, for
  each $i \in \{1, 2\}$, this
  happens iff the run of~$A$ on~$\nu(T_{n_i})$ reaches $q_i$ on the root node $n_i$
  of~$\nu(T_{n_i})$.
  Hence, the condition is equivalent to requiring that there is $q_1, q_2 \in Q$
  such that $\delta((\lambda(n), 0), q_1, q_2) = q$ and, for all $i \in \{1, 2\}$, 
  the run of~$A$ on~$\nu(T_{n_i})$ reaches $q_i$ on the root node $n_i$. By
  definition of $\delta$, this is the case iff 
  the run of~$A$ on the subtree $T_n$ of~$T$ rooted at~$n$ reaches~$q$ on
  the root node~$n$. This concludes the inductive proof of the correctness
  claim.

  This clearly implies that the set captured by the decomposable circuit $C$ is the
  union of the assignment sets $\nu$ such that the run of~$A$ on $\nu(T)$
  reaches a final state at the root, i.e., the assignments $\nu(T)$ for which
  $A$ accepts $\nu(T)$, so the construction is correct.
  
  It remains to show that $C$ is deterministic. The only OR-gates that we
  introduce are the $g^q_n$, and the output gate $g_0$. For a leaf
  node $n \in T$, it is clear from their definition that the $g^q_n$ are deterministic. For an
  internal node $n \in T$ with children $n_1$ and $n_2$,
  the fact that the $g^q_n$ are deterministic is thanks to the determinism of
  the automaton: 
  for every valuation $\nu \in \LVal(T)$, by the inductive invariant, for each
  $i \in \{1, 2\}$, there is
  exactly one $q_i \in Q$ such that $\nu \in g^{q_i}_{n_i}$. Hence, there is
  exactly one $q_1, q_2 \in Q$ such that $\nu \in S(g^{q_1,q_2}_n)$. This
  implies that there could not be a gate $g^q_n$ for some $q \in Q$ such that
  $\nu$ is in the captured set of two of its inputs. For the 
  output gate $g_0$, determinism follows again from the determinism of the
  automaton, as for every leaf valuation $\nu$ of~$T$ the automaton $A$ reaches
  exactly one state on the root of~$\nu(T)$. Thus, $C$ is deterministic.
  This concludes the proof.
\end{proof}

This allows us to recap the proof of Theorem~\ref{thm:facrep}:

\begin{proof}[Proof of Theorem~\ref{thm:facrep}]
  Fix the MSO formula $\phi(X_1, \ldots, X_k)$ on~$\Gamma$-trees, compute the MSO formula
  $\psi(Y)$ on~$\Gamma^{\mathbf{X}}$-trees by Lemma~\ref{lem:simplify}, and the
  Boolean MSO formula $\chi$ on $\overline{\Gamma^{\mathbf{X}}}$-trees by
  Lemma~\ref{lem:boolform}. Rewrite $\chi$ to~$\chi'$ by adding one fresh symbol
  $\bot'$ that can be used to make input trees full, relativizing quantification
  to exclude $\bot'$-nodes from consideration but asserting that they never
  carry the label $(\bot', 1)$, and let
  $\overline{\Gamma'}$ be the resulting alphabet, where $\Gamma' =
  \Gamma^{\mathbf{X}} \cup \{\bot'\}$, the union being disjoint as $\bot'$ is
  fresh. Now, use
  Theorem~\ref{thm:thatcherwright} to compute a $\overline{\Gamma'}$-bDTA
  for~$\chi'$.

  Given the input $\Gamma$-tree $T$, rewrite it in linear time following the
  process of Lemma~\ref{lem:simplify} to a $\Gamma^{\mathbf{X}}$-tree $T'$, and
  complete it with $\bot'$-nodes to a $\Gamma'$-tree $T''$ which is binary and full.
  Now, use
  Proposition~\ref{prp:provenance} to compute a deterministic circuit
  $C$ that captures the assignments of~$A$ on~$T''$ and is a d-DNNF in the
  zero-suppressed semantics. Finally, rewrite $C$ in linear
  time to $C'$ by considering each variable gate~$n$ and
  doing the following:
  
  \begin{itemize}
    \item If $n$ is a leaf of $T''$ which is not in~$T'$ (i.e., it was added
      just to make the tree full), replace $n$ with an OR-gate with no
      inputs. Recalling that $\chi'$ enforces that such nodes are never
      annotated with~$1$ in a valuation, this does not change the captured set
      of the circuit.
      Further, it clearly cannot alter
      decomposability, nor can it alter determinism because the captured set
      $S(g)$ of each gate~$g$
      after this transformation are a subset of the set previously captured by
      gate~$g$.
    \item If $n$ is a leaf of~$T''$ which is in~$T'$, recalling that its label
      $\lambda(n)$ 
      is necessarily in $X_1, \ldots, X_k$, replace the singleton $n$ by $\langle \lambda(n): \mu(n)
      \rangle$. By the condition on~$\mu$, this cannot break decomposability or
      determinism, because it is a bijective renaming of the variable gates.
  \end{itemize}

  Hence, the result $C'$ is a monotone d-DNNF in zero-suppressed
  semantics. We now show that
  it captures the assignments of~$\phi$ on~$T$. For the forward direction,
  consider an assignment $A$ of $\phi$ on~$T$. 
  By Lemma~\ref{lem:simplify}, there is a subset $U$ of leaves
  of~$T'$ such that $\alpha(\{\langle Y : n \rangle \mid n \in U\}) = A$ and $T'$ satisfies $\phi(U)$. By
  Lemma~\ref{lem:boolform}, the leaf valuation $\nu_{U}$ obtained from $U$ is such that
  $\nu_U(T')$ satisfies $\chi$, and clearly if we expand $\nu_{U}$ to a valuation
  of~$T''$ that sets to~$0$ the additional leaves of~$T'$ we know that
  $\nu_{U}(T'')$ satisfies $\chi'$. Hence, by Theorem~\ref{thm:thatcherwright}, we
  know that $A$ accepts $\nu_{U}(T'')$, so by Proposition~\ref{prp:provenance}
  the assignment $U$ corresponding to~$\nu_{U}$
  is captured by~$C$. Now, our rewriting ensures that, as $A = \alpha(\{\langle
  Y: n\rangle \mid n \in U\})$, the
  circuit $C'$ captures~$A$.

  For the backward direction, consider an assignment $A$ captured by the
  monotone d-DNNF $C'$.
  Considering its preimage in~$C$, this means that $C$ captures an
  assignment $U$, i.e., a set of leaves of~$T''$, that are all in
  $T'$ and such that $\alpha(\{\langle Y : n \rangle \mid n \in U\}) = A$.
  Now, by Proposition~\ref{prp:provenance}
  we know that, letting $\nu_{U}$ be the leaf valuation of~$T''$ defined by
  setting the nodes of~$U$ to~$1$ and setting all other nodes to~$0$, the
  automaton $A$ accepts $\nu_{U}(T'')$. By
  Theorem~\ref{thm:thatcherwright}, this implies that $\nu_{U}(T'')$ satisfies
  $\chi'$, hence $\nu_{U}(T')$ satisfies~$\chi$, hence, by
  Lemma~\ref{lem:boolform}, $T'$ satisfies $\psi(U)$, and by
  Lemma~\ref{lem:simplify} we know that $T$ satisfies $\phi(A)$. This concludes
  the correctness proof.
\end{proof}

  \subsection{Proof of MSO Enumeration Results}
  \label{apx:msoenum}

  We now explain formally how our results can be used to re-prove the existing
  result of~\cite{bagan2006mso,kazana2013enumeration}, once we have restricted
  to~$\Gamma$-trees. Note that, unlike what we defined in the main text, this
  result does not only focus on data complexity: the goal is to justify the $O(k \cdot
  \card{T})$ claim for the delay given in the main text.

  \begin{theorem}
  \label{thm:msoenum}
    For any fixed tree alphabet $\Gamma$,
    given an MSO formula~$\phi$ with $k$ free variables and a $\Gamma$-tree $T$,
    we can enumerate the answers to~$\phi$ on~$T$ with the following
    complexities:
    \begin{itemize}
      \item the preprocessing has linear data complexity, i.e., 
        it is in $O(f(\card{\phi}) \cdot \card{T})$ for some fixed function~$f$;
      \item the delay is linear in each produced valuation and independent from the query
    except for~$k$, in particular, it is in $O(k \cdot \card{T})$;
  \item the memory usage is linear in the size of the largest valuation and again
    independent from the query except for~$k$, so again in particular in $O(k \cdot
        \card{T})$.
    \end{itemize}
    If all free variables of~$\phi$ are first-order, the delay and memory
    usage are in $O(k)$.
  \end{theorem}

\begin{proof}
  For the preprocessing phase, we use Theorem~\ref{thm:facrep} to compute in
  linear-time a monotone circuit $C$ which is a d-DNNF in zero-suppressed
  semantics and captures the assignments
  of~$\phi$ on~$T$.
  Note that we have not shown a compatible order for~$C$, but it has no range
  gates, so we know by Remark~\ref{rmk:rangegates}
  that we can apply the results of
  Section~\ref{sec:normal} to the circuit, and the same is immediately true for
  Sections~\ref{sec:multitrees} and~\ref{sec:enum}.

  We further know that this circuit is upwards-deterministic by Claim~\ref{clm:facrepud}
  (see Appendix~\ref{apx:memory}), so we can apply the linear-time preprocessing
  scheme of Theorem~\ref{thm:linmem} as well as its enumeration scheme.
  This runs in delay linear in each
  assignment, which is always in $O(k \cdot \card{T})$, i.e., constant delay
  ($O(k)$) if the size of assignments is constant, which is in
  particular the case if the free variables of~$\phi$ are second-order translations of
  free first-order variables. We then rewrite each assignment (set of
  singletons) to the answer that it represents, in time linear in
  each assignment. The memory usage is linear in each assignment thanks to
  Theorem~\ref{thm:linmem}.
\end{proof}
\end{toappendix}

\subparagraph*{Factorized representations.}
\begin{toappendix}
  \subsection{Factorized Representations}
\end{toappendix}
Our second application is the \emph{factorized representations}
of~\cite{olteanu2015size}, a concise
way to represent database relations \cite{abiteboul1995foundations} by
``factoring out'' common parts. The atomic factorized relations are the empty relation
$\emptyset$, the relation $\langle \rangle$ containing only the empty tuple,
and singletons $\langle A : a \rangle$ where $A$ is an attribute and $a$ is an
element. Larger relations are built using the relational union and Cartesian product operators on
sub-relations with compatible schemas.
For
example, $\langle A_1 : a_1 \rangle \times (\langle A_2 : a_2
\rangle \cup \langle A_2 : a_2' \rangle)$ is a factorized
representation of the relation on attributes $A_1, A_2$ containing 
the tuples $(a_1, a_2)$ and $(a_1, a_2')$.
A \emph{d-representation} is a factorized representation given as a DAG, to
reuse common sub-expressions. We show that
d-representations can be seen as circuits in
zero-suppressed semantics:

\begin{lemmarep}
  \label{lem:factorized}
  For any d-representation $D$, let $C$ be the monotone circuit obtained by replacing $\times$
  and $\cup$ by AND and OR, replacing 
  $\emptyset$ and $\langle \rangle$ by AND-gates and OR-gates with no inputs,
  and keeping singletons as variables. Then all AND-gates of~$C$
  are decomposable, and $S(C)$ (defined as in
  Section~\ref{sec:reducing}) is exactly the
  database relation represented by~$D$.
\end{lemmarep}

\begin{proof}
  The fact that $C$ is decomposable, i.e., a DNNF, is thanks to the requirement
  on d-representations which imposes that gates have a schema, with union always
  having input gates of the same schema, and product always having input gates
  of disjoint schemas. This requirement clearly disallows in particular that
  some singleton has a path to two different inputs to a product gates.
  Note that $C$ is not an augmented circuit, but as it is decomposable, its set
  of assignments $S(C)$ in zero-suppressed semantics (in the sense of
  Definition~\ref{def:zss}, or Lemma~\ref{lem:botup}) is well-defined (recall
  Remark~\ref{rmk:zerosuppressed}).
  The claimed result on~$S(C)$ then follows immediately from Lemma~\ref{lem:botup}. 
\end{proof}

Hence, our results in Theorem~\ref{thm:constant} can be rephrased in terms of
factorized representations:

\begin{theoremrep}
  \label{thm:enumfact}
  The tuples of a deterministic d-representation $D$ over a schema $\calS$ can
  be enumerated with linear-time preprocessing, delay $O(\card{\calS})$, and memory
  $O(\card{\calS}\log\card{D})$.
\end{theoremrep}

\begin{proof}
  Let $D$ be a deterministic d-representation, and let $C$ be the corresponding
  monotone circuit as in the statement of Lemma~\ref{lem:factorized}, such that the set $S(C)$
  captured by $C$ is the relation represented by~$D$: we know that $C$ is
  decomposable.
  The circuit $C$ is not exactly deterministic
  because the determinism requirement of~\cite{olteanu2015size} only requires
  that they are no duplicate tuples in the captured set of the 
  output gate $g_0$. However, it is easy to see that this requirement implies that,
  for every OR-gate $g$, there are no duplicates when computing $S(g)$, unless
  $g$ has no directed path to~$g_0$ or it is ``absorbed'' later in the circuit
  (i.e., we only use its value conjoined with gates capturing~$\emptyset$).
  Hence, we rewrite $C$ to~$C'$ by removing gates with no directed path to~$g_0$, and by
  computing bottom-up in linear time which gates capture exactly~$\emptyset$ (as
  in Lemma~\ref{lem:empprune}), and replace them by OR-gates with no inputs:
  this does not change the set captured by~$C$ (indeed, the sets captured by all
  remaining gates), and $C'$ is still decomposable. Now, it is clear that the
  determinism requirement of~\cite{olteanu2015size} on $S(g_0)$ in~$C$, hence
  on~$S(g_0)$ in~$C'$, imposes that all OR-gates are deterministic, because any
  violation of determinism on a gate~$g$ would imply a duplicate in~$S(g)$,
  hence in $S(g_0)$, following a directed path from $g$ to~$g_0$, and observing
  that the duplicate can never be lost at an OR-gate along the path, or at an
  AND-gate (this uses the fact that no gate captures $\emptyset$). Hence, $C'$
  is a d-DNNF in zero-suppressed semantics such that $S(C')$ is the relation
  represented by~$D$.

  We note that $C'$ does not have a compatible
  order, but again it is decomposable and does not have range gates, so the process in
  Sections~\ref{sec:normal}--\ref{sec:enum} still applies to it (see in
  particular Remark~\ref{rmk:rangegates}), because the process does not
  introduce range gates, and does not use the order except to define the
  semantics of range gates and to guarantee decomposability.
  So we can simply use Proposition~\ref{prp:normalize} to
  compute a normal monotone circuit~$C''$ capturing the same set as $C'$ (it is
  not necessary to apply homogenization because $C'$ already captures tuples of
  the correct weight), we apply Theorem~\ref{thm:multitreeidx}, and last we
  enumerate following Proposition~\ref{prp:enum}. We handle the special cases of~$\ssemp$
  and of circuits capturing~$\emptyset$ like in the proof of
  Theorem~\ref{thm:main} in Appendix~\ref{apx:putting}. Thus, we can enumerate
  the tuples of~$C'$, hence of~$D$, with linear-time preprocessing, delay in
  $O(\card{\calS})$, and memory $O(\card{\calS}\log\card{D})$ as in
  Theorem~\ref{thm:main}.
\end{proof}

Note that the existing enumeration result on factorized representations
(Theorem~4.11 of~\cite{olteanu2015size}) achieves a constant memory bound,
unlike ours. However, this existing result applies only to
deterministic d-representations that are \emph{normal} (Definition~4.6
of~\cite{olteanu2015size}), whereas ours does not assume this.
Normal d-representations are intuitively pruned and
collapsed circuits where \emph{no OR-gate is an input to an
OR-gate}, which avoids, e.g., the need for the constructions of 
Section~\ref{sec:multitrees}. Observe that the circuits that we build for MSO
queries are \emph{not} normal in this sense, so we cannot prove
Theorem~\ref{thm:facrep} directly from Theorem~4.11 of~\cite{olteanu2015size}.

\begin{toappendix}
  \section{Constant-Memory Enumeration for Upwards-Deterministic Circuits}
\label{apx:memory}
Remember that our enumeration results of Theorem~\ref{thm:main} and
Theorem~\ref{thm:constant} use memory in $O(\card{\calO} \log \card{\calI})$, where
$\card{\calI}$ is the size of the input circuit and $\card{\calO}$ is the size of each
output. The factor in~$\card{\calO}$, which is constant for constant-sized
outputs, is obviously difficult to avoid. However, the same is not true of the logarithmic factor in the
input, which comes from the indexing construction on
multitrees of Theorem~\ref{thm:orindex} in Section~\ref{sec:multitrees}. 

In this appendix, we explain how the memory usage of the enumeration phase of
Theorem~\ref{thm:main} and Theorem~\ref{thm:constant} can be improved to
$O(\card{\calO})$, under an
additional hypothesis on the input circuit which allows us to bypass
Theorem~\ref{thm:orindex}. We first present this condition, called
\emph{upwards-determinism}, and claim that enumeration for such circuits can be performed using
memory linear in the size of each valuation (Theorem~\ref{thm:linmem}). Second, we show that the circuits produced for 
MSO enumeration in Theorem~\ref{thm:facrep} are upwards-deterministic. Third, we
prove Theorem~\ref{thm:linmem}.

\subsection{Upwards-Deterministic Circuits}

We define upwards-deterministic circuits in the following way:

\begin{definition}
  A wire $(g, g')$ of~$C$ is \emph{pure} if $g'$ is an OR-gate, or if $g'$ is an
  AND-gate and all its other inputs are 0-valid.
  A gate $g$  is \emph{upwards-deterministic} if $g$ is unsatisfiable or there is at
  most one gate $g'$ such that $(g, g')$ is a pure wire of~$C$.
  We call $C$ \emph{upwards-deterministic} if every AND-gate and OR-gate in~$C$ is
  upwards-deterministic.
\end{definition}

In particular, when a wire $(g, g')$ of a monotone circuit is pure,
it intuitively means that $g'$ evaluates to~$1$ whenever $g$ does, and $S(g)
\subseteq S(g')$ in zero-suppressed semantics. Upwards-determinism imposes that
$g$ is an input to at most one such $g'$.

If we assume upwards-determinism, we can show the analogue of our main
results of Theorem~\ref{thm:main} and Theorem~\ref{thm:constant}, but with
memory usage linear in each output. Namely:

\begin{theorem}
  \label{thm:linmem}
  Given a structured upwards-deterministic d-DNNF $C$ with its v-tree $T$, we
  can enumerate its satisfying assignments with linear-time preprocessing and
  delay and memory usage linear in each valuation.
  Further, for any $k \in \NN$, we can enumerate the satisfying assignments of
  Hamming weight $\leq k$ with preprocessing $O(k^2 \card{C})$ and with delay and
  memory usage in $O(k)$, i.e., constant delay and constant memory.
\end{theorem}

\begin{proofsketch}
  We show that upwards-determinism can be preserved in our preprocessing
  in Sections~\ref{sec:reducing}--\ref{sec:normal}. Once the circuit is normal,
  upwards-determinism ensures that each OR-gate is the input to at most one
  OR-gate, so OR-components in Section~\ref{sec:multitrees} are actually
  reversed trees,
  and we can replace Theorem~\ref{thm:multitreeidx} with a much simpler
  constant-memory indexing scheme.
\end{proofsketch}

The complete proof of Theorem~\ref{thm:linmem} is technical, and presented in
Appendix~\ref{apx:linmemproof}.

\subsection{Upwards-Deterministic Circuits for MSO Enumeration}

We now show the claim that Theorem~\ref{thm:facrep} produces
circuits whose underlying circuit is upwards-deterministic. This implies
the constant memory bound for MSO enumeration in Theorem~\ref{thm:msoenum},
using Theorem~\ref{thm:linmem}.

\begin{claim}
  \label{clm:facrepud}
  Theorem~\ref{thm:facrep} produces upwards-deterministic circuits.
\end{claim}

\begin{proof}
  The input rewriting that we perform in the proof of Theorem~\ref{thm:facrep}
  clearly cannot influence the fact that the circuit is
  upwards-deterministic. Indeed, first, the bijective renaming of inputs clearly has no
  effect. Second replacing some inputs by gates capturing $\emptyset$ ensures
  that the captured set of each gate is a subset of what it was before the
  rewriting: so 
  the set of unsatisfiable gates is a superset of what it was initially,
  and the set of 0-valid gates is a subset of what it was initially, thus
  any violation of upwards-determinism in the initial circuit implies the
  existence of a violation in the original circuit. From this, to show the
  claim for Theorem~\ref{thm:facrep}, it suffices to
  show that the circuits produced in
  Proposition~\ref{prp:provenance} are upwards-deterministic. To show this,
  consider its application to an automaton with state set $A$ and to a
  $\Gamma$-tree $T$, and let $C$ be the resulting circuit.

  In the construction, the only gates that are used as input to multiple
  gates are the $g^q_n$ for $q \in Q$ and $n \in T$ when $n$ is not the root
  of~$T$. Let $n'$ be the parent of~$n$ in~$T$, and assume that $n$ is the
  first child of~$n'$ in~$T$: the proof if $n$ is the second child is
  symmetric. Let $n_2$ be the second child of~$n'$. The gates of~$C$ that have $g^q_n$ as an input are then the
  $g^{q,q_2}_{n'}$ for~$q_2 \in Q$, and the other input to each of them is
  $g^{q_2}_{n_2}$. Now, by determinism of the automaton,
  using the inductive invariant in the proof of
  Proposition~\ref{prp:provenance}, we know that there is exactly one $q_2$
  such that $\semp \in g^{q_2}_{n_2}$, i.e., $g^{q_2}_{n_2}$. Hence, the only outgoing wire
  of~$g^q_n$ which is pure is the one to~$g^{q,q_2}_{n_2}$, so $g^q_n$ does
  not violate upwards-determinism. This concludes the proof.
\end{proof}

\subsection{Proof of Theorem~\ref{thm:linmem}}
\label{apx:linmemproof}

  To show Theorem~\ref{thm:linmem}, we revisit the proofs of
  Sections~\ref{sec:reducing}--\ref{sec:enum}. Specifically:

  \begin{enumerate}
    \item We must show that the preprocessing steps of Sections~\ref{sec:reducing}--\ref{sec:normal}
  preserve upwards-determinism. We must specifically show this for the reduction
      to zero-suppressed semantics (Proposition~\ref{prp:reducing}),
      the homogenization (Proposition~\ref{prp:homogenize}), and the
      normalization (Proposition~\ref{prp:normalize}).
\item We must show that we can replace the use of Theorem~\ref{thm:multitreeidx} in
  Section~\ref{sec:multitrees} by a constant-memory indexing result. To do this,
      we can use the
      assumption that OR-components are reversed trees (i.e., rooted trees,
      where edges are reversed and go from the leaves to the root), because this is guaranteed by
      upwards-determinism on normal circuits. Indeed, a gate $g$ with two
      different children in an OR-component would necessarily be satisfiable
      (because a normal circuit is $\emptyset$-pruned), and its two outgoing
      wires in the OR-component would be pure.
    \item We must show that enumeration in Section~\ref{sec:enum} with the indexes of
  Proposition~\ref{prp:treeidx} uses linear memory.
  \end{enumerate}

  We first show the second point:

  \begin{proposition}
    \label{prp:treeidx}
    Given a reversed tree $T$, we can compute in linear time a data structure allowing
    us to perform the following: given $n \in T$, enumerate in constant delay and
    constant memory the leaves of~$T$ that have a directed path to~$n$.
  \end{proposition}

  \begin{proof}
    We traverse the tree in prefix order in linear time and store at each leaf a
    pointer to the next leaf. We then traverse the tree bottom-up and store,
    for each internal node $n$ of the tree, a pointer to its first leaf in the
    prefix order (i.e., the first leaf that has a directed path to~$n$), and a
    pointer to its last leaf in the prefix order. This can clearly be performed
    in linear time.

    To perform the enumeration, given a node $n$, we jump to its first leaf $n'$,
    remember its last leaf $n''$, and we enumerate the leaves in prefix order from
    $n'$ to $n''$. This process is clearly correct, constant delay, and uses
    only a constant amount of memory.
  \end{proof}

  We next argue for the third point: the process of Section~\ref{sec:enum} takes memory linear
  in each produced valuation (and in particular constant when valuations have
  bounded size). Indeed, the only place in this section where memory usage did
  not satisfy this property was when using the OR-indexes, but the indexes of
  Proposition~\ref{prp:treeidx} only require constant memory, so the overall
  memory usage is linear in the produced valuations.

  We last take care of the first point.
  We first show that rewriting circuits to arity-two can be performed in
  linear-time without breaking upwards-determinism, extending
  Lemma~\ref{lem:arity2}:

  \begin{claim}
    \label{clm:arity2ud}
    Every upwards-deterministic Boolean circuit $C$ can be rewritten
     in linear time to an arity-two circuit $C'$ that is equivalent
    to~$C$ in standard semantics (i.e., captures the same function) and
    that is upwards-deterministic. We can further do so while preserving
    a compatible order~$<$.

    For every $k \in \NN$, every upwards-deterministic monotone
    $k$-augmented Boolean circuit $C$ can be rewritten it in linear time to an
    arity-two monotone $k$-augmented Boolean circuit $C'$ that is equivalent to~$C$ in zero-suppressed
    semantics and is upwards-deterministic. Further, all 
    properties preserved in Lemma~\ref{lem:arity2} are still preserved.
  \end{claim}

  \begin{proof}
    We will use the same construction to show the two claims, and it will essentially
    be the same general construction that we used to show
    Lemma~\ref{lem:arity2}: we rewrite each gate with fan-in greater than 2 to a
    tree of gates of the same type with fan-in two. For this reason, we will not
    argue that the same properties as before are preserved, because this will
    still be true for the same reasons as before. To preserve
    upwards-determinism, we will simply be more specific about the way in which
    we construct each tree.

    We first preprocess the circuit once to compute which gates are
    0-valid. This can clearly be performed in linear time, as in
    Lemma~\ref{lem:pruned}.

    Whenever we wish to rewrite a gate $g$ with input gates $g_1, \ldots, g_n$,
    with $n > 2$, the tree of gates of the same type that we introduce will be
    linear (i.e., as unbalanced as possible). Specifically, we remove the wires
    from $g_i$ to~$g$ for $1 \leq i \leq n$, we introduce
    gates $g'_i$ of the same type as~$g$ for $1 < i < n-1$,
    we set the inputs of~$g$ to be $g_1$ and $g'_1$, the inputs of each $g'_i$
    for $1 < i < n - 2$ to be $g_{i+1}$ and $g'_{i+1}$, and the inputs of $g'_{n-2}$ to
    be $g_{n-1}$ and~$g_n$.
    This ensures that all gates have arity-two, and that the circuit is
    equivalent.

    We now impose a constraint on the order in which the input gates $g_1,
    \ldots, g_n$ should be considered: we require that all gates that are
    0-valid are enumerated first, so they are attached as high in the tree
    as possible.

    The only thing to show is that upwards-determinism is preserved.
    The new gates, i.e., the $g'_i$ introduced for each gate~$g$, cannot introduce a violation of
    upwards-determinism, because they have only one outgoing wire (to
    $g'_{i-1}$, or to~$g$). Hence, it suffices to consider outgoing wires for gates
    of the rewritten circuit $C'$ that stand for gates of the input circuit
    $C$, i.e., using our terminology above, it suffices to consider the wires
    from the $g_i$ to the~$g'_j$, or to~$g$. It clearly suffices to show that,
    whenever such a wire is pure, then the corresponding wire $(g_i, g)$ is pure
    in~$C$. Indeed, this implies that any violation of upwards-determinism
    in~$C'$ on a gate~$g'$ (which also exists in~$C$) would imply a violation of
    upwards-determinism on~$g'$ in~$C$.

    Hence, let us consider a wire $(g', g'')$ in~$C'$ where $g'$ exists in~$C$,
    let $g$ be the gate for which $g''$ was introduced: observe that the wire $(g', g)$
    exists in~$C$, and that $g$ and $g''$ have the same type, in fact possibly we
    have $g'' = g'$. Let us assume that $(g', g'')$ is pure in~$C'$, and show
    that it is pure in~$C$. There are four possibilities:

    \begin{itemize}
      \item \emph{The gate $g$ is an OR-gate.} In this case, the wire is pure in~$C$, and there is nothing to show.
      \item \emph{The gate $g$ is an AND-gate and all its inputs are 0-valid
        in~$C$.} In
        this case, the wire is pure in~$C$, and there is nothing to show.
      \item \emph{The gate $g$ is an AND-gate and only one of its inputs $g^*$ is
        not 0-valid in~$C$.} In this case, the only incoming pure wire of~$g$ in~$C$
        is $(g^*, g)$, and under our assumption that $(g', g'')$ is pure in~$C'$ we must show that $g'
        = g^*$. From the construction
        we know that $g^*$ is still 0-valid in~$C'$, so we know
        that $g^*$ was enumerated last in the inputs of~$g$, so it is attached
        to the lowest node in the tree of~$C'$ introduced for~$g$. As $g^*$ is
        still not 0-valid in~$C'$, we then know that $g$ is
        not 0-valid in~$C'$ and none of the $g'_i$
        is 0-valid (because there is a path from the gate
        $g^*$, which is not 0-valid, to all these gates that goes only via AND-gates). So if
        the wire $(g', g'')$ is pure, it must be the case that $g''$ is the lowest node in the
        tree of~$C'$, and the other input to~$g''$ must be 0-valid so we must
        have $g' = g^*$ which is what we wanted to show.
      \item \emph{The gate $g$ is an AND-gate and at least two of its inputs are
        not 0-valid in~$C$.} In this case, similarly to the above reasoning, $g$
        is not 0-valid in $C$ and none of the $g'_i$ are 0-valid
        in~$C$, Further, as two inputs that are not 0-valid were enumerated last,
        the lowest node in the tree has two inputs that are not 0-valid. Hence, in fact,
        this case cannot occur under our assumption that the wire $(g', g'')$ is
        pure in~$C'$.
    \end{itemize}

    This concludes the proof.
  \end{proof}

  We then show:

  \begin{claim}
    \label{clm:reducingud}
    The construction of Proposition~\ref{prp:reducing} preserves
    upwards-determinism.
  \end{claim}

  \begin{proof}
    The construction first completes the input circuit using
    Lemma~\ref{lem:completion}, and then computes its
    monotonization. We argue that monotonization on a decomposable circuit cannot break
    upwards-determinism. Indeed, it does not change which gates are
    0-valid, it does not change the type of AND-gates or OR-gates except to
    introduce AND-gates with no inputs, so it does not change which wires are
    pure; and further it cannot make any gate unsatisfiable which
    wasn't unsatisfiable. We thus focus on completion.
    
    The completion
    construction in the proof of Lemma~\ref{lem:completion} first
    rewrites the circuit to an arity-two circuit $C$, which does not break
    upwards-determinism by Claim~\ref{clm:arity2ud}. Then it adds range gates as
    children to some AND-gates and rewrites inputs to OR-gates by AND-gates of
    the original gates and some fresh range gates. We explain why this does not break
    upwards-determinism.

    It is clear that, in the new circuit $C'$, the wires going out of a new range
    gate or out of a new AND-gate cannot violate upwards-determinism, because
    these gates are used as input to only one gate. So it suffices to consider
    the wires $(g', g)$ going out of gates $g'$ in~$C'$ that correspond to gates
    that already existed in~$C$. There are two cases: either $g$ is an AND-gate
    of~$C'$ that already existed in~$C$, or $g$ is an AND-gate introduced when
    rewriting an OR-gate $g''$ of~$C$.

    In the first case, we show that if the wire $(g', g)$ is pure in~$C'$, then
    it was already pure in~$C$. But this is immediate: if the wire is pure, then
    all other inputs to~$g$ in~$C'$ are 0-valid, and then from our
    rewriting it is clear that all inputs of~$g$ in~$C$ (which are a subset of
    those in~$C'$) were already 0-valid.

    In the second case, as~$g''$ was an OR-gate of~$C$, the wire $(g, g'')$ was
    necessarily pure in~$C$.

    This allows us to conclude the proof. Indeed, assume by way of contradiction
    that there is a gate $g$ of~$C'$ that violates upwards-determinism. By our
    initial reasoning, $g$ is necessarily a gate that already exists in~$C$.
    Further, $g$ captures a non-empty set in~$C'$, and by our construction we
    know that the same is true of~$g$ in~$C$. Now, let $g_1 \neq g_2$ be the
    gates of~$C'$ such that the wires $(g, g_1)$ and $(g, g_2)$ are pure
    in~$C'$. Let $g_1'$, $g_2'$ be the gates that correspond to~$g_1$ and $g_2$
    in~$C$, i.e., $g_i' = g_i$ if $g_i'$ exists in~$C$, and otherwise $g_i'$ is
    the OR-gate of $C$ for which the AND-gate $g_i$ was introduced. Our
    construction clearly ensures that $g_1' \neq g_2'$: indeed, our construction
    ensures that the gate $g$ cannot have a wire both to a fresh AND-gate
    of~$C'$ and to the original OR-gate (indeed no gates at all have wires to
    the original OR-gates), and $g$ cannot have a wire to two new AND-gates
    introduced for the same OR-gate (as we create one AND-gate for each input).
    Now, our previous claim ensures that the wires $(g, g_1')$ and $(g, g_2')$
    are pure in~$C$, so $g$ witnesses that $C$ is not upwards-deterministic,
    contradicting our assumption and concluding the proof.
  \end{proof}

  We then show the claim for Proposition~\ref{prp:normalize}.
  The construction of Lemma~\ref{lem:arity2} extends thanks to
  Claim~\ref{clm:arity2ud}. It is straightforward that
  Lemma~\ref{lem:collapsed} does not break upwards-determinism. Indeed, wires to
  AND-gates that are collapsed are necessarily pure because they have only one
  input, so collapsing the gates cannot cause a gate to have more than one
  outgoing pure wire. Further, adding exits is not problematic, because wires to
  exits were already to OR-nodes, so already pure, and each exit has exactly one
  outgoing wire.

  The $\emptyset$-pruning process of Lemma~\ref{lem:empprune} preserves upwards-determinism. Indeed,
  any gate in the output existed with the same type in the input, it is
  0-valid in the output iff it is 0-valid in the input, all gates in the output
  are satisfiable but were already satisfiable in the input, every wire in
  the output existed in the input, and it is not hard to see that if a wire in
  the output is pure then is also pure in the input: indeed, the inputs to
  AND-gates are unchanged, and changing the inputs to OR-gates is unproblematic
  because all their incoming wires are always pure.
  
  What must be shown is that upwards-determinism 
  is preserved by the pruning construction of Lemma~\ref{lem:pruned}. In this lemma, the
  actual process of $\semp$-pruning is unproblematic for similar reasons as for
  $\emptyset$-pruning: note that removing input gates to AND-gates that are
  0-valid cannot cause any of the other input wires to become pure.
  The crux of the matter is to
  show that Proposition~\ref{prp:homogenize2} preserves upwards-determinism. 
  This also takes care of proving the extension of
  Proposition~\ref{prp:homogenize}. Hence, we claim:

  \begin{claim}
    \label{clm:homogud}
    If the input $C$ to Proposition~\ref{prp:homogenize2} is upwards-deterministic,
    then its output $C'$ also is.
  \end{claim}

  \begin{proof}
    Remember that the construction in the proof of
    Proposition~\ref{prp:homogenize2} first rewrites the circuit to arity-two
    with Lemma~\ref{lem:arity2}, which does not break upwards-determinism thanks
    to Claim~\ref{clm:arity2ud}; so we
    let $C$ be the arity-two version of the input circuit, which is
    upwards-deterministic.
    
    The construction then produces $C'$ by introducing, for each gate $g$ of the
    original circuit~$C$, gates of the form $g^{=i}$ for $0 \leq i \leq k$ and
    $g^{>k}$, as well as gates of the form $g_j^{=i}$ and $g_{i,j}^{>k}$,
    $g_{i,j}^{>k,1}$, $g_{i,j}^{>k,2}$, which we call \emph{fresh gates} of~$C'$.
    We will define the \emph{original gate}
    $\omega(g)$ of a gate $g$ of~$C$ as follows:

    \begin{itemize}
      \item if $g$ is of the form $g_0^{=i}$ or $g_0^{>k}$, then $\omega(g)
        \colonequals g_0$
      \item if $g$ is a fresh gate created for a AND-gate $g_0$ with two inputs
        in~$C$, then $\omega(g) \colonequals g_0$.
    \end{itemize}
    
    It is clear that fresh gates $g$ in~$C'$ cannot violate
    upwards-determinism, because in the construction any such gate $g$ is used as
    input to only one gate in~$C'$, specifically, a gate whose original gate is
    the same as that of~$g$.
    So it suffices to check upwards-determinism for gates of~$C'$ which are not
    fresh gates, i.e., wires $(g', g)$ of~$C'$ where $g'$ is not fresh, so that
    in particular $\omega(g') \neq \omega(g)$, and by construction $(\omega(g'),
    \omega(g))$ is a wire of~$C$.

    We will show the following claim (*): for every wire $(g_0', g_0)$ of~$C$
    such that $g_0'$ is an AND-gate or an OR-gate, when
    considering every wire $(g', g)$ of~$C'$ such that $\omega(g) = g_0$ and
    $\omega(g') = g_0'$, then (i) for each choice of~$g'$, at most one $g$ is such that
    the wire is pure, and (ii)
    if one such wire is pure then $(g_0', g_0)$ is also pure in~$C$. This claim implies that
    $C'$ is upwards-deterministic. Indeed, assume to the contrary that $C'$ is
    not upwards-deterministic, then it has an AND- or OR-gate $g'$ which is not
    fresh, is satisfiable, and has two pure wires $(g', g_1)$ and $(g, g_2)$ with $g_1 \neq
    g_2$. The construction then ensures that $\omega(g')$ is an AND-gate or an
    OR-gate and that $(\omega(g'), \omega(g_1))$ and $(\omega(g'), \omega(g_2))$
    are wires of~$C$. Further, by the properties of~$C'$, the set $S(g')$
    captured by $g'$ in~$C'$ is a subset of the set $S(\omega(g'))$
    of~$\omega(g')$ in~$C$, so $\omega(g')$ is satisfiable.
    By (i), we know that we must have
    $\omega(g_1) \neq \omega(g_2)$,
    and by (ii) these two wires are pure in~$C$, so $\omega(g')$ is
    not upwards-deterministic in~$C$, a contradiction.
    Hence, it suffices to show claim (*).
    
    Let us show claim (*) by considering all possible wires $(g_0', g_0)$
    of~$C$:

    \begin{itemize}
      \item If $g_0$ is an OR-gate, then the wire $(g_0', g_0)$ is always pure so (ii)
        is vacuous. Further, for each gate $g'$ of~$C'$ with $\omega(g') =
        g_0'$, there is exactly one gate $g$ of~$C'$ with $\omega(g) = g_0$
        such that the wire $(g', g)$ is in~$C'$, so (i) holds.
      \item If $g_0$ is an AND-gate, then:
        \begin{itemize}
          \item If $g_0$ has no inputs, then there are no wires to consider
            so (i) and (ii) are vacuous.
          \item If $g_0$ has one input then the wire $(g_0', g_0)$ is always pure
            so (ii) is vacuous, and (i) holds for the same reasons as for
            OR-gates.
          \item If $g_0$ has two inputs, let $g_0''$ be the input of~$g_0$
            in~$C$ which is different from~$g_0'$, i.e., the inputs of~$g_0$
            in~$C$ are $g_0'$ and $g_0''$. Observe that in the construction, for any wire $(g', g)$
            of $C$ with $\omega(g) = g_0$ and $\omega(g') = g_0'$, the gate $g$
            is always a fresh AND-gate with two inputs, and its other input is a
            non-fresh gate $g''$ such that $\omega(g'') = g_0''$. Hence, 
            the wire $(g',
            g)$ of~$C'$
            is pure only if $g''$ is 0-valid in~$C'$.
            Recalling the properties of~$C'$, remember that this can only happen
            if $g''$ is the gate $(g_0'')^{=0}$ and if $g_0''$ is 0-valid
            in~$C$, so we have shown point (ii). Further, observe from the
            construction that the only
            such wires $(g', g)$ in~$C'$ are:
            \begin{itemize}
              \item For $i \in \{0, \ldots, k\}$, the wire from $(g_0')^{=i}$ to
                $(g_0)_i^{=i}$, whose other input is $(g_0'')^{=0}$.
              \item The wire from $(g_0')^{>k}$ to $(g_0)^{>k,1}$, whose other
                input is $(g_0'')^{=0}$.
            \end{itemize}
            So indeed, for each choice of~$g$, there is at most one pure wire,
            so (i) holds too.
        \end{itemize}
    \end{itemize}

    We have thus established claim (*), which concludes the proof.
  \end{proof}

  With the above, we have finished the proof of Theorem~\ref{thm:linmem}.

\end{toappendix}

\section{Conclusion}
\label{sec:conclusion}
We have studied how to enumerate satisfying valuations of circuits, under
the structuredness, decomposability, and determinism conditions introduced in
AI: we have shown that enumeration can be performed with
linear preprocessing and delay linear in each valuation (so constant delay
for valuations of constant Hamming weight). We have given two example applications of this
result: factorized databases, and an independent proof of the MSO query enumeration
results of~\cite{bagan2006mso,kazana2013query}. Beyond these applications,
however, our method implies
efficient enumeration results for all problems studied in knowledge compilation,
when they can be compiled to structured d-DNNFs (refer back to the Introduction for examples).

A natural question is whether our constructions can be extended for other tasks, e.g.,
computing the $i$-th valuation \cite{bagan2006mso,BaganDGO08}; managing updates
on the structure \cite{losemann2014mso}; or enumerating
valuations in order of weight, or in lexicographic order: this latter problem is open for MSO \cite[Section~6.1]{Segoufin14}
though results are known for factorized representations following an
\mbox{f-tree}~\cite{bakibayev2013aggregation}. Another direction is to 
strengthen our result to constant-memory enumeration on all d-DNNF circuits, or
lift some hypotheses on the input circuits.
We also intend to study a practical implementation, which we
believe to be realistic since our construction only performs simple and modular
transformations on the input circuits, and has no hidden large constants.

\subparagraph*{Acknowledgements.}
This work was partly funded by the French ANR Aggreg project, by 
the CPER Nord-Pas de Calais/FEDER DATA Advanced data science and technologies
2015-2020, by the PEPS JCJC INS2I 2017 CODA, and by the Télécom ParisTech
Research Chair on Big Data and Market Insights.

\vfill
\pagebreak

\bibliography{main}

\end{document}